\documentclass[letterpaper, 11 pt]{article} 
\usepackage{amsmath,amsfonts,amssymb,color,amsthm}
\usepackage{nccmath}
\usepackage{epsfig}
\usepackage{psfrag}
\usepackage[T1]{fontenc}
\usepackage{algorithm}
\usepackage{algpseudocode}
\usepackage[dvipsnames]{xcolor}
\usepackage{array}
\usepackage{epstopdf}
\usepackage{cite}
\usepackage{footmisc}
\usepackage{mathtools}
\usepackage{bbm}
\usepackage{bm}
\usepackage{comment}
\usepackage{mathtools}
\usepackage{epstopdf}
\usepackage{tikz}
\usetikzlibrary{automata, positioning, arrows.meta}
\usepackage{relsize}
\usetikzlibrary{shapes,arrows}
\usepackage{cite}
\usepackage{epstopdf}
\usepackage{tcolorbox}
\definecolor{winered}{rgb}{0.5,0,0}
\usepackage{scalerel}
\usepackage{xcolor}
\usepackage[colorlinks]{hyperref}
\usepackage{tcolorbox}
\tcbset{colback=blue!5!white,colframe=blue}

\newcommand{\mc}{\mathcal}

\newcommand{\normlr}[1]{\left\lVert#1\right\rVert}
\newcommand{\stackref}[2]{
\readlist*\mylist{#1}
\stackrel{\mbox{\footnotesize\foreachitem\x\in\mylist[]{\ifnum\xcnt=1\else,\fi\eqref{\x}}}}{#2}
}

\newcommand{\abs}[1]{\lvert #1 \rvert}
\newcommand{\abslr}[1]{\left\lvert#1\right\rvert}

\newcommand{\tr}[1]{\text{trace}\left( #1 \right)}

\AtBeginDocument{%
  \hypersetup{
    citecolor=blue,
    linkcolor=winered
  }
}

\DeclarePairedDelimiter\ceil{\lceil}{\rceil}
\DeclareMathOperator*{\argmin}{\arg\!\min}

\DeclareMathOperator{\Var}{Var}
\DeclarePairedDelimiterX{\norm}[1]{\lVert}{\rVert}{#1}
\DeclareMathOperator{\E}{\mathbb{E}}

\newtheorem{problem}{Problem}

\newtheorem{theorem}{Theorem}

\newtheorem{lemma}{Lemma}

\date{}
\usepackage{cite}
\usepackage[margin=1in]{geometry}

\title{\LARGE Outlier-Robust Linear System Identification Under Heavy-tailed Noise}

\author{Vinay Kanakeri and Aritra Mitra} % Use a specific footnote mark

\begin{document}
\maketitle
\footnotetext[1]{The authors are with the Department of Electrical and Computer Engineering, North Carolina State University. Email: {\tt \{vkanake, amitra2\}@ncsu.edu}.}
\begin{abstract}%
 We consider the problem of estimating the state transition matrix of a linear time-invariant (LTI) system, given access to multiple independent trajectories sampled from the system. Several recent papers have conducted a non-asymptotic analysis of this problem, relying crucially on the assumption that the process noise is either Gaussian or sub-Gaussian, i.e., ``light-tailed". In sharp contrast, we work under a significantly weaker noise model, assuming nothing more than the existence of the fourth moment of the noise distribution. For this setting, we provide the first set of results demonstrating that one can obtain sample-complexity bounds for linear system identification that are nearly of the same order as under sub-Gaussian noise. To achieve such results, we develop a novel robust system identification algorithm that relies on constructing multiple weakly-concentrated estimators, and then boosting their performance using suitable tools from high-dimensional robust statistics. Interestingly,  our analysis reveals how the kurtosis of the noise distribution, a measure of heavy-tailedness, affects the number of trajectories needed to achieve desired estimation error bounds. Finally, we show that our algorithm and analysis technique can be easily extended to account for scenarios where an adversary can arbitrarily corrupt a small fraction of the collected trajectory data.  Our work takes the first steps towards building a robust statistical learning theory for control under non-ideal assumptions on the data-generating process. 
\end{abstract}

\section{Introduction}
\label{sec:Intro}
Given the empirical success of reinforcement learning in various complex tasks spanning video games to robotics, there has been a recent growth of interest in understanding the performance of feedback control systems when the model of the system is \emph{unknown}~\cite{hu2023toward, tsiamis2023statistical}. To mitigate uncertainty in the model, one natural strategy is to first use data generated by the system to learn the system parameters - a task known as \emph{system identification.} Subsequently, using the learned system model, one can appeal to either certainty-equivalent or robust control. For such a data-driven approach to yield the desired stability and performance guarantees, it is essential to quantify how data sampled from the system can help reduce the uncertainty in the underlying dynamics. This is particularly important for reliable operation in safety-critical applications (e.g., self-driving cars) when one integrates data-driven approaches into the feedback control loop. In this context, a growing body of work has drawn upon tools from learning theory and high-dimensional statistics to characterize the number of samples needed to accurately estimate the system parameters, given access to noisy data. While the results in this space have provided a fine-grained understanding of how the nature of the dynamical system (stable vs. unstable) shapes the sample-complexity bounds for system identification, all such results have been derived under somewhat idealistic assumptions on the data-generating process. In particular, the process noise exciting the dynamics is assumed to be either Gaussian or sub-Gaussian, i.e., ``light-tailed", an assumption that may not hold for real-world environments. This begs the following questions. \emph{Under less favorable circumstances when the noise is heavy-tailed or even adversarial, can we still provide a finite-sample analysis of system identification? If so, is there any hope of recovering similar estimation error bounds as achievable under sub-Gaussian noise?} 
In this paper, we provide the first rigorous examination of the above questions for the task of linear system identification.  

More precisely, we consider a linear time-invariant (LTI) system: 
$x_{t+1}= A x_t + w_t$, where $x_t \in \mathbb{R}^d$ and $w_t \in \mathbb{R}^d$ are the state of the system and process noise at time $t$, respectively, and $A \in \mathbb{R}^{d \times d}$ is the unknown state transition matrix. Given access to $N$ independent trajectories sampled from this system, the goal is to construct an estimate $\hat{A}$ of $A$, and characterize the corresponding sample-complexity bounds within the probably approximately correct (PAC) framework. In other words, we wish to precisely quantify the number of trajectories needed to achieve a prescribed estimation accuracy $\varepsilon$ with a confidence level specified by failure probability $\delta.$ The main contribution of this work is to offer the first PAC bounds for this setting under the assumption that the noise process $\{w_t\}$ admits a finite fourth moment, \emph{and nothing more}. Interestingly, we show that with a suitably designed robust estimator of $A$, one can (almost) recover bounds known in the literature under the significantly stronger assumption of sub-Gaussian noise.\footnote{For a precise definition of sub-Gaussianity, see Chapters 2 and 3 of~\cite{vershynin2018high}. Note that Proposition 2.5.2. of this reference guarantees the existence of all finite moments of a sub-Gaussian distribution.} Before elaborating further on our results, we briefly summarize the relevant literature on system identification. 

\noindent \textbf{Related Work.} Linear system identification is a fundamental problem in control theory that finds applications in time-series forecasting, finance, and reinforcement learning. Classical treatments of this problem primarily focus on asymptotic results~\cite{lai_asymp, ljung1987theory}. Our interest, however, is in a more recent strand of literature that aims to provide a finer understanding by deriving non-asymptotic bounds on the amount of data that needs to be collected from the system to meet desired performance guarantees. To our knowledge, the first results of this kind were obtained in~\cite{rantzer2018} for scalar LTI systems. In follow-up work~\cite{dean2020sample}, the results were extended to vector (potentially unstable) LTI systems for the multi-trajectory setting, i.e., when multiple independent trajectories are available as data to the learner. The multi-trajectory setting we consider here is akin to that in~\cite{dean2020sample} and has also appeared in various other works~\cite{zheng2020non, xing2020linear, xin2022learning, tu2024learning}. When data is collected from a single trajectory, the analysis becomes much more challenging since such data is no longer independent and identically distributed (i.i.d.), but rather temporally correlated. For stable and marginally stable LTI systems, finite-sample results for the single-trajectory case were derived in~\cite{simchowitz1},~\cite{oymak2019non}, and~\cite{jedra2022finite}. For unstable LTI systems, results under single-trajectory data were obtained in~\cite{sarkar2019near}. 

Variations of the basic linear system identification problem involve scenarios where the system state is not fully observed, and the system is only excited by noise; see the work of~\cite{tsiamis2019finite} in this regard. Other variants include the problem of sparse system identification, considered in~\cite{fattahi2019learning} and~\cite{sun2020finite}. For a detailed discussion of the latest results on system identification, we refer the reader to the excellent tutorials~\cite{tsiamis2023statistical} and~\cite{ziemann2023tutorial}. Despite the wealth of literature that has emerged on the topic in recent years, the analysis in each of the papers mentioned above hinges crucially on leveraging concentration bounds for Gaussian or sub-Gaussian noise distributions. One notable exception is the work of~\cite{faradonbeh2018finite}, where the authors consider a noise model weaker than ones admitting sub-Gaussian tails. Nonetheless, the sub-Weibull noise model in~\cite{faradonbeh2018finite} ensures the existence of all finite moments of the noise distribution. This leads to the following question we investigate in our work: 

\textit{Can we derive finite sample bounds for linear system identification under heavy-tailed noise distributions that admit no more than the fourth moment?}

The recent survey paper~\cite{tsiamis2023statistical} identifies this as an open question. We provide an answer in the affirmative via the following contributions.

$\bullet$ \textbf{Problem Formulation.} Our study is motivated by an interesting observation made in~\cite{tsiamis2023statistical}. The authors note that for a heavy-tailed noise model where, for instance, $\mathbb{E}[\Vert w_t \Vert^4] < \infty$ but $\mathbb{E}[\Vert w_t \Vert^p] = \infty$ for some finite $p >4$, while the ordinary least squares (OLS) estimator might still be optimal in expectation under i.i.d. data, it is no longer optimal w.r.t. its dependence on the failure probability $\delta$~\cite{oliveira2016lower}. In particular, it fails to achieve the optimal logarithmic dependence of $\log(1/\delta)$ for all distributions within the aforementioned heavy-tailed noise class. In the context of heavy-tailed linear system identification, we examine for the first time whether such an optimal $\log(1/\delta)$ dependence can be reinstated. 

$\bullet$ \textbf{Novel Algorithm.} In practice, it may not be possible to ascertain ahead of time whether the noise is sub-Gaussian or heavy-tailed. As such, we would ideally like to have a system identification algorithm that is \emph{agnostic} to the nature of the noise and yields similar guarantees under both light- and heavy-tailed distributions. As discussed earlier, the OLS estimator fails in this regard. This motivates us to develop a novel algorithm titled \texttt{Robust-SysID} in Section~\ref{sec:Algo}. Our main idea is to first construct multiple OLS estimators of $A$ by suitably partitioning the collected trajectories into buckets. To ``boost" the performance of such weakly concentrated estimators, we employ the notion of a geometric median w.r.t. the Frobenius norm. While similar ideas have been pursued for robust mean estimation~\cite{minsker2015geometric}, we show how they can be also employed for system identification.  

$\bullet$ \textbf{Matching Sub-Gaussian Rates under Heavy-tailed Noise.} Existing analyses for system identification exploit various concentration tools for sub-Gaussian and sub-exponential distributions. Unfortunately, our noise model precludes the use of such tools, necessitating a new proof technique. To illustrate some of the key ideas that show up in our analysis, we consider a scalar setting in Section~\ref{sec:scalar}. Our main result for the scalar case, namely Theorem~\ref{thm:scalar_case}, recovers the \emph{exact same error bound as under sub-Gaussian noise}; in particular, we are able to achieve the desired $\log(1/\delta)$ dependence on the failure probability. Unlike the analogous sub-Gaussian result, however, the number of trajectories needs to scale with the \emph{kurtosis} of the noise distribution, i.e., the ratio of the fourth moment to the square of the variance. While this requirement captures the effect of heavy-tailed noise, whether it is fundamental is an open question. The extension to the vector setting requires much more work to control the smallest eigenvalue of the empirical covariance matrix. Our main result for this case, namely Theorem~\ref{thm:vector_case}, once again nearly recovers the same error bound as reported in~\cite{dean2020sample} under sub-Gaussian noise, up to an extra multiplicative $O(d)$ term. \emph{To our knowledge, these are the first results to demonstrate that one can (almost) match sub-Gaussian error rates under heavy-tailed noise for system identification.}  

$\bullet$ \textbf{Robustness to Outlier Trajectories.} Finally, in Theorem~\ref{thm:corruption}, we show that our algorithm and analysis technique can be seamlessly extended to account for the scenario where an adversary can arbitrarily corrupt a small fraction $\eta$ of the trajectories. Our result in this context is consistent with those for robust mean estimation with adversarial outliers~\cite{lugosi2021robust}. 

Overall, by drawing on ideas from robust statistics, we take the first steps towards building the foundations of data-driven control under non-ideal (yet more realistic) assumptions on the data-generating process. While we focus on system identification in this paper, we anticipate that our ideas  will find broader applicability to more complex feedback control problems under uncertainty.

\noindent \textbf{Notation.} Given a positive integer $n \in \mathbb{N}$, we define the shorthand $[n] \triangleq \{1, 2, \ldots, n\}.$ For a vector $w \in \mathbb{R}^d$, we will use $w^{\top}$ to denote its transpose, and $w(i)$ to represent its $i$-th component. Unless otherwise specified, $\Vert \cdot \Vert$ will be used to denote the Euclidean norm for vectors and spectral norm for matrices. Given a matrix $M \in \mathbb{R}^{d \times d}$, we will use $\Vert M \Vert_F$ to denote its Frobenius norm. Finally, we will use $c, C, c_1, c_2, \ldots$ to represent universal constants that may change from one line to another. 
\newpage
\section{Problem Formulation}
\label{sec:Prob_form}
Consider an uncontrolled discrete-time linear time-invariant (LTI) system of the following form:
\begin{equation}
x_{t+1}= A x_t + w_t,
\label{eqn:sys_model}
\end{equation}
where $x_t \in \mathbb{R}^d$ and $w_t \in \mathbb{R}^d$ are the state of the system and process noise at time $t$, respectively, and $A \in \mathbb{R}^{d \times d}$ is the a priori \emph{unknown} state transition matrix. Without loss of generality, we assume that $x_0=0$. We further assume that the noise sequence $\{w_t\}$ is a zero-mean, independent and identically distributed (i.i.d.) stochastic process satisfying the following second- and fourth-moment bounds: 
\begin{equation}
\mathbb{E}[w_t w^{\top}_t] = \sigma^2 I_d, \hspace{3mm} \mathbb{E}[(w_t(i))^4] = \tilde{\sigma}^4, \forall i \in [d], \forall t \geq 0. 
\label{eqn:noise_model}
\end{equation}

\noindent \textbf{Data Collection.} Suppose we have access to $N$ independent trajectories of the system \eqref{eqn:sys_model}, each of length $T$. Such trajectories can be generated by rolling out the dynamics for $T$ time-steps, and then resetting the system to the zero initial condition after each rollout. Let us use $\mc{D}^{(i)}$ to denote the trajectory data $\{x^{(i)}_t\}_{1\leq t \leq T+1}$ collected during the $i$-th rollout, where $i \in [N].$ Using the collective data set $\mc{D} = \bigcup_{i \in [N]} \mc{D}^{(i)}$, the goal of a learner is to obtain an estimate of the system matrix $A$. Formally, our problem of interest can now be stated as follows.

\begin{problem} \label{prob:robust_sys_id} Consider the system in \eqref{eqn:sys_model} and the noise model in \eqref{eqn:noise_model}. Fix an accuracy parameter $\varepsilon > 0$ and a failure probability $\delta \in (0,1).$ Given the data set $\mc{D}$, construct an estimator $\hat{A}$ of $A$, and characterize its sample-complexity $N_S(\varepsilon, \delta, C_A, C_w),$ such that with probability at least $1- \delta,$ we have 
$ \Vert \hat{A} - A \Vert \leq \varepsilon,$ provided $N \geq N_S.$ Here, $C_A$ and $C_w$ are constants that depend on the system matrix $A$, and the noise parameters $\sigma, \tilde{\sigma}$, respectively. 
\end{problem}

Several comments are now in order regarding our problem formulation.
\begin{enumerate}
\item[1.] The key departure of our problem setting from existing finite-time results on linear system identification stems from the generality of the assumptions we make on the noise process $\{w_t\}$. In particular, existing work on this topic has either assumed ``light-tailed" Gaussian or sub-Gaussian noise. The only notable exception we are aware of in this regard is the work in~\cite{faradonbeh2018finite}, where the authors consider a noise process with sub-Weibull distribution. Although more general than sub-Gaussian distributions, all finite moments of a sub-Weibull distribution exist, as shown in~\cite{vladimirova2020sub}. In sharp contrast, the assumptions we make on the noise process in~\eqref{eqn:noise_model} require nothing more than the existence of the fourth moment of the noise distribution. 

\item[2.] To build intuition regarding our results, let us consider the well-studied setting where the noise process is Gaussian with variance $\sigma^2.$ Given $N$ independent trajectories, the ordinary least squares (OLS) estimator yields the following guarantee in this scenario~\cite{dean2020sample}:
\begin{equation}
\Vert \hat{A} - A \Vert \leq c_1 \sqrt{  \frac{d \log(1/\delta)}{ \lambda_{\min}(G_T) N}} \hspace{2mm} \textrm{holds with probability at least $1-\delta$}, 
\label{eqn:benchmark}
\end{equation}
provided $N \geq c_2 d \log(1/\delta).$ Here, $c_1, c_2 >0$ are suitable universal constants, and $G_T := \sum_{t=0}^{T-1} A^t (A^{\top})^t$. Our \textbf{goal} is to understand whether, and to what extent, similar guarantees can be recovered under the significantly more general noise model we consider in this paper. In particular, we ask: \emph{Is it possible to retain the mild logarithmic dependence on the failure probability $\delta$ in \eqref{eqn:benchmark}}?  This is particularly relevant when one seeks high-probability guarantees. 

\item[3.] To focus on answering the above question, we consider a system model with no inputs. Nonetheless, under the standard assumption of controllability, our techniques can be easily extended to account for a somewhat more general system of the form $x_{t+1}= A x_t + B u_t + w_t$, where $B$ is an unknown input matrix, and $u_t$ is the control input at time $t$. In a similar spirit, to isolate the effect of heavy-tailed noise, here, we do not pursue other natural extensions pertaining to partial observability, measurement noise, single-trajectory data, and nonlinear dynamics. While each of these generalizations are certainly interesting avenues for future work, they are orthogonal to the subject of this paper. 
\end{enumerate}

In the sequel, we will show that it is indeed possible to recover bounds of the form in~\eqref{eqn:benchmark}. Arriving at such bounds will, however, require an algorithmic technique different from the standard OLS approach. Furthermore, as we will elaborate later in the paper, we cannot appeal to the existing proof techniques for linear system identification that rely heavily on concentration properties of light-tailed distributions. This is all to say that the ``simple" model in \eqref{eqn:sys_model} is sufficiently interesting in its own right. 

\section{Robust System Identification Algorithm}
\label{sec:Algo}
In this section, we will develop our proposed algorithm called \texttt{Robust-SysID} that enables system identification in the face of heavy-tailed noise. Later, in Section~\ref{sec:corruption}, we will see that a minor tweak to this algorithm suffices to accommodate the presence of arbitrarily corrupted adversarial data. In other words, \emph{we will establish that \texttt{Robust-SysID} is not only robust to heavy-tailed noise, but also to adversarial outliers.} Our algorithm has three main components that we outline below. 

\textbf{Step 1: Bucketing.} In the first step, we partition the $N$ data sets into $K$ buckets denoted by $\mc{B}_1, \ldots, \mc{B}_K$, such that each bucket contains $M$ independent trajectories; here, for simplicity, we have assumed that $N=MK$. The choice of $K$ is crucial for our final bounds and will be specified later in the statement of our main results. 

\textbf{Step 2: Local Estimation per Bucket.} In the second step, we use the trajectories within each bucket to construct an OLS estimator per bucket. To make this idea precise, fix a bucket $j \in [K].$ We use the last two samples of each trajectory within $\mc{B}_j$ to construct the OLS estimator $\hat{A}_j$ for bucket $j$:
\begin{equation}
\hat{A}_j = \argmin_{\theta \in \mathbb{R}^{d \times d}} \sum_{i \in \mc{B}_j} \Vert x^{(i)}_{T+1} - \theta x^{(i)}_T\Vert^2.
\label{eqn:OLS}
\end{equation}

\textbf{Step 3: Boosting.} In the last step, we fuse the ``weak" estimates obtained from each bucket to create a more powerful estimator for $A$. Specifically, we leverage the notion of a geometric median for matrices to construct $\hat{A}$ as follows:
\begin{equation}
\hat{A}= \texttt{Med}(\hat{A}_1, \ldots, \hat{A}_K) := \argmin_{\theta \in \mathbb{R}^{d \times d}} \sum_{j \in [K]} \Vert \theta - \hat{A}_j \Vert_F.
\label{eqn:GM}
\end{equation}
In the above step, the geometric median $\hat{A}$ is computed with respect to the Frobenius norm since it induces an inner-product space on the space of all matrices in $\mathbb{R}^{d \times d}.$ In turn, this guarantees the existence of $\hat{A}$ as defined in~\eqref{eqn:GM}~\cite{minsker2015geometric}. At a high level, we note that our algorithmic strategy is inspired by the popular ``median of means" device from robust statistics. While such ideas have been explored in the past for robust mean estimation, we employ them here for the first time in the context of linear system identification. In the subsequent sections, we will discuss the performance guarantees of \texttt{Robust-SysID}, and sketch out the main steps in the analysis, while highlighting the challenges that arise in the way. 

\section{Warm Up: The Scalar Case}
\label{sec:scalar}
We start by analyzing a scalar version of the system in  \eqref{eqn:sys_model} since it captures much of the challenges posed by heavy-tailed noise. Therefore, analyzing \texttt{Robust-SysID} for this case provides us with valuable insights into the nature of the bounds to be expected in the more challenging vector setting. Our main result on the performance of \texttt{Robust-SysID} for the scalar case is stated in the following theorem,  where we define $a$ to be the scalar counterpart of $A$ from \eqref{eqn:sys_model}, and $g_T = \sum_{t=0}^{T-1} a^{2t}$.

\begin{theorem} \label{thm:scalar_case} 
Consider the scalar version of the system in \eqref{eqn:sys_model} and the noise assumptions in \eqref{eqn:noise_model}. With probability at least $1-\delta$, the following bound holds for the output $\hat{a}$ of \texttt{Robust-SysID}:
\begin{equation}
    \abs{\hat{a} - a} \leq C \sqrt{\frac{\log(1/\delta)}{N g_T}}, \hspace{2mm} \textrm{provided} \hspace{1mm} K = \ceil{c_1 \log(1/\delta)}, M \geq c_2 (\tilde{\sigma}^4/\sigma^4), \hspace{1mm} \textrm{and} \hspace{1mm} N=MK.
    \label{eqn:scalar_err_bound}
\end{equation}
\end{theorem}
Before providing a proof sketch of the above result, some remarks are in order.

\noindent \textbf{Discussion.} We note that the error bound in \eqref{eqn:scalar_err_bound} matches the one obtained by the standard OLS estimator under Gaussian noise, as indicated in \eqref{eqn:benchmark}. This reveals the robustness of our algorithm to a general heavy-tailed noise process. Theorem~\ref{thm:scalar_case} also specifies the design parameters of \texttt{Robust-SysID}, namely the number of buckets $K$, and the number of samples per bucket $M$. Since $N=MK$, we note from \eqref{eqn:scalar_err_bound} that the number of trajectories depends on the kurtosis of the noise process - a dependence not observed under Gaussian noise. Although we are uncertain whether a dependence on the kurtosis is inevitable, it is, however, meaningful as it captures the heaviness of the tail. 

\textbf{Proof Sketch for Theorem~\ref{thm:scalar_case}.} In what follows, we sketch the main ideas in the proof of Theorem~\ref{thm:scalar_case}. The detailed proof is deferred to Appendix~\ref{app:scalar_proof}. The main hurdle in our analysis is that we can no longer leverage concentration bounds for sub-Gaussian and sub-exponential distributions that have appeared in prior works. Nonetheless, we start by deriving bounds for the OLS estimator from each bucket. Accordingly, the OLS estimator $\hat{a}_j$ for the $j$th bucket can be expressed as 
\begin{equation}
   \hat{a}_j = a + \frac{ \sum_{i \in \mc{B}_j} x_T^{(i)}w_T^{(i)}}{ \sum_{i \in \mc{B}_j} \big( x_T^{(i)}\big)^2}.
   \label{eqn:ols_scalar}
\end{equation}
We bound the numerator and the denominator of the error term separately  and then combine them by applying an union bound. In this regard, the following lemmas provide key results. 

\begin{lemma} \label{lemma:scalar_num} (\textbf{Scalar numerator upper bound})  Fix a bucket $j \in [K]$. With probability at least $1 - p/2$, the following holds:
$$\abslr{\sum_{i \in \mc{B}_j} x_T^{(i)}w_T^{(i)}} \leq c\sigma^2\sqrt{g_T M/p}.$$ 
\end{lemma}

To prove the above result, one can start by noting that due to the i.i.d nature of the trajectories, $\Var(\sum_{i \in \mc{B}_j} x_T^{(i)}w_T^{(i)}) = M \Var(x_T^{(1)}w_T^{(1)})$, where $M = \abs{\mc{B}_j}$, and $\Var(Z)$ is used to represent the variance of a real-valued random variable $Z$. Next, observe that for each individual term, $\Var(x_T^{(1)}w_T^{(1)}) = \mathbb{E}[(x_T^{(1)})^2] \mathbb{E}[(w_T^{(1)})^2] = \sigma^2 g_T \times \sigma^2 = \sigma^4 g_T$; here, we exploited the fact that $x_T^{(1)}$ and $w_T^{(1)}$ are independent, and $\mathbb{E}[x_T^{(1)}] = \mathbb{E}[w_T^{(1)}]=0$. The rest follows from a straightforward application of Chebyshev's inequality. Next, the following lemma provides a lower bound on $\sum_{i \in \mc{B}_j} \big( x_T^{(i)}\big)^2$, by exploiting the existence of the fourth moment of $w^{(i)}_T$.

\begin{lemma} \label{lemma:scalar_den} (\textbf{Scalar denominator lower bound}) Fix a bucket $j \in [K]$. With probability at least $1 - p/2$, the following holds: 
$$\sum_{i \in \mc{B}_j} \big( x_T^{(i)}\big)^2 \geq \sigma^2 g_T M/2, \hspace{2mm} \textrm{provided} \hspace{2mm}  M \geq (c/p) (\tilde{\sigma}^4/\sigma^4). $$ 
\end{lemma}

\begin{proof} \label{proof:scalar_den} Due to the i.i.d nature of the trajectories, we have $\Var\big(\sum_{i \in \mc{B}_j} \big( x_T^{(i)}\big)^2\big) = M \Var\big(\big( x_T^{(1)}\big)^2\big).$ For clarity of notation, let us drop the superscript in the rest of the proof. Since $\Var\big( x_T^2\big) \leq \E\left[x_T^4\right],$ it suffices to bound the fourth moment of $x_T$. Under the zero initial condition, observe that $ x_T = \sum_{t = 0}^{T-1} a^tn_t$, where we have defined $n_t \triangleq w_{T - (t+1)}$ for brevity. The fourth moment of $x_T$ can be expressed as follows:
\begin{equation*}
    \begin{aligned}
    \E[x_T^4] &= \E\left[ \left(\sum_{t = 0}^{T-1} a^tn_t\right)^2 \left(\sum_{s = 0}^{T-1} a^sn_s\right)^2\right] \\
    & = \E\left[\left( \underbrace{\sum_{t = 0}^{T-1} a^{2t}n_t^2}_{T_1} + \underbrace{\sum_{t' \neq t = 0}^{T-1}a^{t'}a^tn_tn_{t'}}_{T_2}\right) \left( \underbrace{\sum_{s = 0}^{T-1} a^{2s}n_s^2}_{T_3} + \underbrace{\sum_{s' \neq s = 0}^{T-1}a^{s'}a^sn_sn_{s'}}_{T_4}\right)\right].
    \end{aligned}
\end{equation*}
In the above display, as a result of the noise process being i.i.d. with zero mean, only the terms that contribute either a fourth power or a product of squares survive the expectation, causing the cross terms $T_1 \times T_4$ and $T_2 \times T_3$ to vanish. Furthermore, in $T_2 \times T_4$, only terms of the form $n_t^2n_{t'}^2$ survive, yielding
\begin{equation*}
\begin{aligned}
    \E[x_T^4] = \sum_{t=0}^{T-1} a^{4t} \E[n_t^4] + 3\sum_{t' \neq t = 0}^{T-1}a^{2(t+t')}\E[n_t^2n_{t'}^2] 
    = \sum_{t=0}^{T-1} a^{4t} \tilde{\sigma}^4 + 3\sum_{t' \neq t = 0}^{T-1}a^{2(t+t')}\sigma^4 
     \overset{(a)} \leq 3 (g_T)^2 \tilde{\sigma}^4.
\end{aligned}
\end{equation*}
In the above steps, (a) follows from the definition of $g_T$ in Theorem~\ref{thm:scalar_case}, and $\tilde{\sigma}^4 \geq \sigma^4$ due to Jensen's inequality. Now applying Chebyshev's bound with the above result, we have for any $t > 0:$
\begin{equation*}
    \mathbb{P} \left( \left|\sum_{i \in \mc{B}_j} \big(x^{(i)}_T\big)^2 - \E\left[\sum_{i \in \mc{B}_j} \big(x^{(i)}_T\big)^2 \right] \right| \geq t \right) \leq \frac{3M(g_T)^2 \tilde{\sigma}^4}{t^2}.
\end{equation*}
Setting the R.H.S. of the above inequality to $p/2$, we get $t =  g_T\tilde{\sigma}^2\sqrt{6M/p}$. Notice that $\E\left[\sum_{i \in \mc{B}_j} \big(x^{(i)}_T\big)^2 \right] = Mg_T\sigma^2$, giving us the following with probability at least $ 1 - p/2$:
\begin{equation*}
     \sum_{i \in \mc{B}_j} \big(x^{(i)}_T\big)^2 \geq g_T\left( \sigma^2 M - \tilde{\sigma}^2\sqrt{6M/p} \right).
\end{equation*}
In the above display, setting the R.H.S $ \geq g_T\sigma^2 M/2$ and solving for $M$ completes the proof.
\end{proof}
Combining the results from Lemma~\ref{lemma:scalar_num} and Lemma~\ref{lemma:scalar_den}, and using an union bound, we have that when $M \geq (c_1/p) (\tilde{\sigma}^4/\sigma^4)$, the following holds with probability at least $1 - p$:
\begin{equation}
    \abs{\hat{a}_j - a} \leq c_2\sqrt{\frac{1}{p M g_T}}.
    \label{eqn:scalar_bound_pre_boost}
\end{equation}
Note that the failure probability $p$ appears polynomially (and not logarithmically) in the above bound. 
\newpage
\textbf{The Role of Boosting.} In \eqref{eqn:scalar_bound_pre_boost}, set $p=1/4$, and let $\varepsilon = c_2({p M g_T})^{-1/2}.$ In the scalar case, note that $\hat{a}$ in \eqref{eqn:GM} is simply the standard (scalar) median of $\{\hat{a}_1, \ldots, \hat{a}_K\}.$ By the property of the median, observe that the ``bad" event $\{ | \hat{a}-a| > \varepsilon\}$ implies $\{ \sum_{j \in [K]} Y_j \geq K/2 \}$, where $Y_j$ is an indicator random variable of the event  $\{ | \hat{a}_j-a| > \varepsilon\}$. Using the fact that each of the $Y_j$'s are i.i.d. random variables in $\{0,1\}$  satisfying $\mathbb{E}[Y_j] \leq p = 1/4$, we can use Hoeffding's inequality to infer that 
$$ \mathbb{P}\left(\{ | \hat{a}-a| > \varepsilon\}\right) \leq \exp(-K/8) \leq \delta,$$
when $K=\ceil{8 \log(1/\delta)}.$ Using this expression for $K$ in $\varepsilon= c_2({p M g_T})^{-1/2}$ and noting that $M=N/K$, we arrive at the bound in~\eqref{eqn:scalar_err_bound}. This completes the proof sketch for Theorem~\ref{thm:scalar_case}. In simple words, for the median estimate $\hat{a}$ to deviate from $a$ beyond our desired error tolerance $\varepsilon$, at least half of the OLS estimates from the buckets must also deviate by $\varepsilon$. Although the failure probability of each one of such (independent) events is at most $1/4$, asking $K/2$ of such events to occur \emph{simultaneously} diminishes the overall failure probability, thereby ``boosting" the quality of $\hat{a}$. With this intuition in mind, we now proceed to analyze the vector case in the following section. 

\section{The Vector Case}
In this section, we demonstrate how \texttt{Robust-SysID} addresses Problem~\ref{prob:robust_sys_id}. We also discuss some of the unique challenges posed by the vector setting compared to the scalar case from the previous section. The following theorem captures our main result. 
\begin{theorem} \label{thm:vector_case} (\textbf{Main Result}) 
Consider the system in \eqref{eqn:sys_model} and the noise assumptions in \eqref{eqn:noise_model}. With probability at least $1-\delta$, the following bound holds for the output $\hat{A}$ of \texttt{Robust-SysID}:
\begin{equation}
    \norm{\hat{A} - A} \leq C d^{3/2}\sqrt{\frac{\log(1/\delta)}{N \lambda_{\min}(G_T)}}, \hspace{1mm} \textrm{when} \hspace{1mm} K = \ceil{c_1 \log(1/\delta)}, M \geq c_2 d^2 C_A C_w, \hspace{1mm} \textrm{and} \hspace{1mm} N=MK,
    \label{eqn:vector_err_bound}
\end{equation}
\begin{equation}
 \textrm{where} \hspace{2mm} {C_A} \triangleq \left(\frac{\sum_{t=0}^{T-1} \norm{A^t}^2}{\lambda_{\min}(G_T)}\right)^2, \hspace{2mm} \textrm{and} \hspace{2mm} C_w \triangleq \frac{\tilde{\sigma}^4}{\sigma^4}.
    \label{eqn:vector_sample_comp}
\end{equation}
\end{theorem}
The detailed proof of the above result is provided in Appendix~\ref{app:vector_proof}. Before sketching the main ingredients in the analysis, we discuss the implications of Theorem~\ref{thm:vector_case}. 

\textbf{Discussion.} Comparing \eqref{eqn:benchmark} and \eqref{eqn:vector_err_bound}, we note that \texttt{Robust-SysID} recovers the logarithmic dependence on the failure probability. To our knowledge, this is the first result to provide such a guarantee for the general heavy-tailed noise model considered in this work. That said, we note from \eqref{eqn:vector_err_bound} that our error bound, and the requirement on the number of trajectories, both suffer from an extra multiplicative factor of $O(d)$ relative to the Gaussian benchmark in~\eqref{eqn:benchmark}. Furthermore, unlike the scalar case in~\eqref{eqn:scalar_err_bound}, the requirement on the number of trajectories in the vector case exhibits an additional dependency on the system matrix $A$ via the constant $C_A$ in \eqref{eqn:vector_sample_comp}. 

\textbf{Proof Sketch for Theorem~\ref{thm:vector_case}.} Analogous to the scalar case, the proof of Theorem~\ref{thm:vector_case} first involves deriving bounds for the OLS estimators of each bucket. To simplify our analysis, we whiten the vector $x_T^{(i)}$ and define $z_T^{(i)} = \Sigma_x^{-1/2}x_T^{(i)}$, where $\Sigma_x = \E[x_T^{(i)}\big(x_T^{(i)}\big)^{\top}]$. Under this definition, it suffices to individually bound $\normlr{\sum_{i \in \mc{B}_j} w_T^{(i)}\big(z_T^{(i)}\big)^{\top}}$ and $\lambda_{\min}{\left(\sum_{i \in \mc{B}_j} z_T^{(i)}\big(z_T^{(i)}\big)^{\top}\right)}$ as shown in \cite{dean2020sample}, \cite{matni2019tutorial}. The following lemmas provide key results in this regard.

\begin{lemma} \label{lemma:vector_num} (\textbf{Vector numerator upper bound}) Fix a bucket $j \in [K]$. With probability at least $1 - p/2$, the following holds:
$$\normlr{\sum_{i \in \mc{B}_j} w_T^{(i)}\big(z_T^{(i)}\big)^{\top}} \leq c_1 d \sqrt{\sigma^2 M/p}.$$
\end{lemma}

\textbf{Challenges in Analysis.} Let us discuss some of the challenges that arise in the proof of the above result by outlining potential proof strategies, and their limitations for our setting. In \cite{dean2020sample}, the authors derive a similar result under Gaussian noise by exploiting variational properties of the spectral norm along with covering arguments. This was made possible due to the availability of sub-Gaussian and sub-exponential tail bounds with logarithmic dependence on the error probability, which, in turn, help in controlling certain covering numbers. Clearly, without the logarithmic factor, as is the case with heavy-tailed noise, using such an approach would lead to a prohibitive exponential dependence on the dimension $d$ due to the covering number. An alternative strategy to bound the norm of $\sum_{i \in \mc{B}_j} w_T^{(i)}\big(z_T^{(i)}\big)^{\top}$ is to bound each scalar entry of this matrix by invoking the analysis from Section~\ref{sec:scalar}. However, this approach fails to provide the bound in Lemma~\ref{lemma:vector_num} as it would involve union bounding over $d^2$ elements, leading to an additional dimension factor. 

\textbf{Our methods.} In light of the above discussion, we identify two different approaches. In the first approach, we define a new variance statistic for a random square matrix $X$ as $\text{var}(X) \triangleq \E[\norm{X - \E[X]}_F^2]$. Such a definition using the Frobenius norm leverages independence in the sense that $\text{var}(\sum_i X_i) = \sum_i \text{var}(X_i)$ for independent matrices $\{X_i\}$. Since $\text{var}(X)$ as defined above is a scalar, one can use the standard Markov's inequality in this case. The second approach exploits a matrix version of Markov's inequality proposed in~\cite{ahlswede2001strong}. It turns out that both approaches lead to exactly the same bounds in Lemma~\ref{lemma:vector_num} and Lemma~\ref{lemma:vector_den}. Our next result controls the smallest eigenvalue of the (whitened) empirical covariance matrix.

\begin{lemma} \label{lemma:vector_den} (\textbf{Vector denominator lower bound}) For each bucket $j$, with probability at least $1 - p/2$, the following holds:
$$\lambda_{\min}\left(\sum_{i \in \mc{B}_j} z_T^{(i)}\big(z_T^{(i)}\big)^{\top}\right) \geq M/2, \hspace{2mm} \textrm{provided} \hspace{2mm} M \geq c (d^2/p) C_A C_w,$$
where $C_A$ and $C_w$ are as defined in~\eqref{eqn:vector_sample_comp}. 
\end{lemma}
The key step in the proof of the above result involves bounding the trace of the matrix $\E[(x_Tx_T^{\top})^2]$; here, we have dropped the superscript $(i)$ for clarity of exposition. By exploiting the i.i.d. and zero-mean properties of the noise process $\{w_t\}$, the next result helps considerably simplify the expression for $\E[(x_Tx_T^{\top})^2]$. 
\begin{lemma}
    \label{lemma:vec_den_biquadratic_terms} Define $n_t \triangleq w_{T-(t+1)}.$
    Given the system in \eqref{eqn:sys_model} and the noise assumptions in \eqref{eqn:noise_model}, we have
    \begin{equation}
    \resizebox{0.8\hsize}{!}{$
    \begin{aligned}
\E\left[(x_Tx_T^{\top})^2\right] &= \sum_{t=0}^{T-1} \E\left[(A^tn_tn_t^{\top}(A^t)^{\top})^2\right] 
         + 2\sum_{s\neq t=0}^{T-1} \E\left[ A^tn_tn_t^{\top}(A^t)^{\top} A^sn_sn_s^{\top}(A^s)^{\top}\right]\\ 
         &+ \sum_{s\neq t=0}^{T-1} \E\left[ A^tn_tn_s^{\top}(A^s)^{\top} A^sn_sn_t^{\top}(A^t)^{\top}\right]. 
         \nonumber 
\end{aligned}$}
    \end{equation}
\end{lemma}
We use the above result in tandem with various trace inequalities to establish Lemma~\ref{lemma:vector_den}. Combining the results from Lemma~\ref{lemma:vector_num} and Lemma~\ref{lemma:vector_den} immediately provides a guarantee for the OLS estimate from each bucket, much like in \eqref{eqn:scalar_bound_pre_boost}. For boosting, we employ an argument similar to the scalar case, along with properties of the geometric median from~\cite{minsker2015geometric}. 

\textbf{Avenues for improvement.} We discuss the sources of the extra $O(d)$ factor (relative to the sub-Gaussian noise case) in our error-bound of~\eqref{eqn:vector_err_bound}. To invoke the results from~\cite{minsker2015geometric} for the geometric median, we need to work with the Frobenius norm. As such, we use the inequality $\norm{\hat{A_j} - A}_F \leq \sqrt{d} \norm{\hat{A_j} - A}$, costing us an extra $\sqrt{d}$ factor in the boosting step. This could be avoided if it were possible to provide guarantees for robust matrix aggregation directly w.r.t. the spectral norm. The other $\sqrt{d}$ factor comes from Lemma~\ref{lemma:vector_num} for which we used variants of Markov's inequality. One might hope that a more powerful concentration tool can lead to a tighter bound. Perhaps the most relevant result in this context is provided by Theorem 5.48 in \cite{vershynin2010introduction}, which concerns bounding the expected value of the spectral norm of a matrix with heavy-tailed rows. Applied to our setting, we obtain with probability at least $1 - p$, $\normlr{\sum_{i \in \mc{B}_j} w_T^{(i)}\big(z_T^{(i)}\big)^{\top}} \leq c_1 \sqrt{d \sigma^2 M/p} + c_2\sqrt{m \log(d)/p}$, where $m = \E\left[\max_{k\in[d]} \normlr{\sum_{i \in \mc{B}_j} w_T^{(i)}(k)z_T^{(i)}}^2\right]$. Notice that if we could commute the $\max$ and $\E[\cdot]$ operators in the definition of $m$, we would be able to shave off a $\sqrt{d}$ factor from the bound of Lemma~\ref{lemma:vector_num}. However, such an operation is not valid in general. On the other hand, if we upper-bound the $\max$ by summing over $k \in [d]$, we end up with the same bound as in Lemma~\ref{lemma:vector_num}. Although tighter concentration bounds are available for the maxima of sub-Gaussian random variables, we are unaware of analogous bounds under our noise assumptions. Thus, it remains an open problem to ascertain whether our current bounds can be further improved.
\vspace{-3mm}
\section{System Identification under Adversarial Corruptions}
\label{sec:corruption}
In this section, we show that our prior developments concerning \texttt{Robust-SysID} can be extended to account for adversarial corruption in conjunction with heavy-tailed noise. To make this precise, we consider the \emph{strong-contamination} attack model from the robust statistics literature~\cite{lugosi2021robust}, where an adversary can \emph{arbitrarily} corrupt a small fraction $\eta \in [0, 1/2)$ of the data. In our context, we allow the adversary to contaminate up to $\eta N$ number of trajectories in the data set $\mc{D}$. We have the following result for this setting. 

\begin{theorem} \label{thm:corruption} (\textbf{Robustness to adversarial corruptions}) Consider the strong-contamination model described above. With probability at least $1-\delta$, the following bound holds for the output $\hat{A}$ of \texttt{Robust-SysID} when $\eta < 0.5/(c_1d^2C_AC_w), K \geq \ceil{c_2 \log(1/\delta) + c_3 \eta N},$ and $ M \geq c_4 d^2C_AC_w$: 
\begin{equation}
    \norm{\hat{A} - A} \leq C d^{3/2}\left(\sqrt{\frac{\log(1/\delta)}{N \lambda_{\min}(G_T)}} + \sqrt{\frac{\eta}{ \lambda_{\min}(G_T)}}\right).
    \label{eqn:cprruption_err_bound}
\end{equation}
\end{theorem}
Comparing the above result with the case without adversarial corruptions from Theorem~\ref{thm:vector_case}, notice that the error bound in \eqref{eqn:cprruption_err_bound} recovers the bound in \eqref{eqn:vector_err_bound} up to an additive $O(\sqrt{\eta})$ term; this is consistent with analogous results for robust mean estimation in~\cite{lugosi2021robust}. Therefore, Theorem~\ref{thm:corruption} shows that \texttt{Robust-SysID} can effectively counter adversarial corruption by carefully designing the number of buckets, and leveraging the inherent robustness of the geometric median. To gain intuition, consider the worst-case scenario where the adversary corrupts $\eta N$ buckets by corrupting exactly one trajectory in each such bucket. To keep the median well-concentrated in this case, we need the number of \emph{uncorrupted} buckets to be in the order of $\log(1/\delta)$ as shown in \eqref{eqn:vector_err_bound}, which can be ensured with $O(\eta N)$ extra buckets. This explains the requirement on $K$ in Theorem~\ref{thm:corruption}, which, in turn, imposes bounds on $N$ and $\eta$. 
% This is possible because the geometric median operator is inherently robust to arbitrary corruptions. Since the corrupted buckets are unknown a priori, the number of trajectories in each bucket $(M)$ must also satisfy the bound in \eqref{eqn:vector_err_bound}. Given that $N = MK$, this requirement imposes a bound on $N$ and $\eta$ as shown in the above theorem. 
Interestingly, similar constraints on $N$ and $\eta$ are \emph{not} required for robust mean estimation~\cite{lugosi2021robust}. This difference can be attributed to the fact that, unlike mean estimation, sys-ID requires a minimum number of trajectories $M$ per bucket to ensure that the empirical covariance matrix in each uncorrupted bucket is well-behaved. 

\section{Conclusion}
\label{sec:conclusion}
System identification is one of the central components of many algorithms that aim to learn and control dynamical systems. Particularly, understanding the sample complexity for learning the parameters of an unknown system in finite time is of immense practical value. To this end, several works have studied the problem of finite sample analysis of system identification; however, these works make restrictive assumptions on the noise sequence which could render them impractical. To bridge this gap, we considered the problem of linear system identification under heavy tailed noise distributions that admit no more than the fourth moment. We showed that, even under such a general assumption, one could provide finite sample bounds that nearly match the ones obtained under the more restrictive Gaussian or sub-Gaussian noise assumptions.

For this purpose, we proposed the \texttt{Robust-SysID} algorithm which carefully divides the collected trajectories into buckets, computes OLS estimates for each bucket, and finally fuses the OLS estimates using the geometric median. By leveraging simple concentration tools and integrating techniques from robust statistics, particularly the geometric median, we derived strong performance guarantees for our algorithm. While our bounds incur a dimensional dependence factor $d$ compared to Gaussian and sub-Gaussian settings, our work lays the foundation for advancing the analysis and design of algorithms under heavy-tailed noise. To our knowledge, better algorithms and/or analysis techniques that could match the Gaussian and sub-Gaussian bound under the general heavy-tailed noise assumption made in this work is an open problem. 

There are several open questions concerning the heavy-tailed and adversarial noise processes considered in this paper. We list some of them below.

\begin{enumerate}
\item Our main results in Theorems~\ref{thm:vector_case} and~\ref{thm:corruption} feature certain dependencies on the state transition matrix, the dimension $d$ of the system, and the kurtosis of the noise distribution. Deriving information-theoretic lower bounds for the noise models considered in our paper will shed light on the tightness of our bounds. 

\item A natural next step is to consider the single-trajectory setting. For this case, extending the ideas from robust statistics used in this paper is an interesting open direction, especially for unstable systems. 

\item The proofs of our main results exploited the linearity of the dynamical system on several occasions. As such, it would be interesting to see how such proofs can be extended to the non-linear setting~\cite{ziemann2022single, foster2020learning}. 

\item Beyond system identification, we are interested in studying data-driven model-based and model-free control under more general noise processes. 
\end{enumerate}

Finally, it would be interesting to derive data-dependent bounds, such as those in~\cite{pananjady2020instance}. We hope that our work sparks interest in some of the above questions. Overall, we believe that integrating tools from high-dimensional robust statistics into data-driven control-theoretic problems could provide novel algorithms and analysis techniques that help us understand the fundamental limits of operations for many engineering problems. 

\section{Acknowledgments} The authors would like to thank Professor Arvind K. Saibaba for valuable discussions on the topic of this paper. 

\bibliographystyle{unsrt}
%\setcitestyle{authoryear,open={square} and close={square}}
\bibliography{References}

\newpage

\appendix

\section{Preliminary Results and Facts} \label{app:priliminary_results}
In this section, we provide some of the key results and facts that will aid us in the analysis provided in the subsequent appendices. We use the following notation: $ A \succeq 0$ means that $A$ is positive semi-definite, $A \succ 0$ means that $A$ is positive definite, $A \preceq B$ means that $B - A$ is positive semi-definite, and $A \not\preceq B$ means that $B - A$ is not positive semi-definite.
%, and finally, for a symmetric positive semi-definite matrix $A$, $\abslr{A} = \sqrt{A^2}$.
\begin{lemma}
    \label{lemma:trace_inequalities}
    The following inequalities hold for any symmetric positive semi-definite matrices $A$ and $B$, where $A, B \in \mathbb{R}^{d \times d}$:
    \begin{align}
        \tr{AB} &\leq \tr{A}\tr{B} \label{eqn:tr_ab_eq_tra_trb}\\
        \tr{A} &\leq d\lambda_{\max}(A) \label{eqn:tr_a_d_lambda}.
    \end{align}
\end{lemma}
\begin{proof}
    The proof of \eqref{eqn:tr_ab_eq_tra_trb} follows from von Neumann's trace theorem (\cite{horn2012matrix}, Theorem 8.7.6) which tells us that for symmetric positive semi-definite matrices $A$ and $B$, we have $\tr{AB} \leq \sum_{i = 1}^d \lambda_i(A) \lambda_i(B)$ where $\lambda_1(A) \geq \lambda_2(A) \geq \ldots \geq \lambda_d(A)$ and $\lambda_1(B) \geq \lambda_2(B) \geq \ldots \geq \lambda_d(B)$ are ordered eigenvalues of $A$ and $B$, respectively. Notice that $\sum_{i = 1}^d \lambda_i(A) \lambda_i(B) \leq \left(\sum_{i = 1}^d \lambda_i(A)\right) \left(\sum_{j = 1}^d \lambda_j(B)\right) = \tr{A}\tr{B}$.

    The proof of \eqref{eqn:tr_a_d_lambda} follows from the fact that $\tr{A} = \sum_{i = 1}^d \lambda_i(A) \leq d \lambda_{\max}(A)$ where $\lambda_{\max}(A)$ is the largest eigenvalue of $A$.
\end{proof}

\begin{lemma}
    \label{lemma:trace_inequalities_rank_1}
    The following hold for any symmetric positive semi-definite rank 1 matrix $A \in \mathbb{R}^{d \times d}$ and a symmetric positive semi-definite matrix $B \in \mathbb{R}^{d \times d}$:
    \begin{align}
        \tr{AB} &\leq \tr{A}\lambda_{\max}(B) \label{eqn:tr_ab_tr_a_lambda_b}\\
        (\tr{A})^2 &= \tr{A^2} \label{eqn:tr_a_sq}.
    \end{align}
\end{lemma}
\begin{proof}
    Similar to the proof of Lemma \ref{lemma:trace_inequalities}, \eqref{eqn:tr_ab_tr_a_lambda_b} follows from von Neumann's trace theorem (\cite{horn2012matrix}, Theorem 8.7.6). We have $\tr{AB} \leq \sum_{i = 1}^d \lambda_i(A) \lambda_i(B) \leq \lambda_{\max}(B)\sum_{i = 1}^d \lambda_i(A) = \tr{A}\lambda_{\max}(B)$, where $\lambda_{\max}(B)$ is the maximum eigenvalue of $B$. 

    The proof of \eqref{eqn:tr_a_sq} follows from the fact that $(\tr{A})^2 = \left(\sum_{i = 1}^d \lambda_i(A)\right)^2 = \left(\lambda_{\max}(A)\right)^2$ where the last equality is due to the  rank 1 nature of $A$ which allows a single non-zero eigenvalue. Furthermore, notice that $\left(\lambda_{\max}(A)\right)^2 = \tr{A^2}$ as matrix $A^2$ is also rank 1 with a single non-zero eigenvalue $\left(\lambda_{\max}(A)\right)^2$. This completes the proof. 
\end{proof}

\begin{lemma} (\textbf{Markov's inequality for matrices})
    \label{lemma:matrix_markov}
    Let $X$ be a random matrix such that $X \succeq 0$ almost surely and with expectation $\mathbb{E}[X]$, and let $A \succ 0$, then
    \begin{equation}
    \label{eqn:matrix_markov}
        \mathbb{P}\left( X \not\preceq A \right) \leq \tr{\E[X]A^{-1}}. 
    \end{equation}
\end{lemma}
For the proof of the above result, see Theorem 12 from \cite{ahlswede2001strong}.

% \begin{lemma} (\textbf{Chebyshev's inequality for matrices})
% \label{lemma:matrix_chebyshev}
%     Let $X$ be a symmetric random matrix with expectation $\E[X]$ and variance $\Var(X) = \E[(X - \E[X])^2]$, and let $\abslr{X - \E[X]} = \sqrt{(X - \E[X])^2}$. For any $A \succ 0$, we have
%     \begin{equation}
%         \label{eqn:matrix_chebyshev}
%         \mathbb{P}\left( \abslr{X - \E[X]} \not\preceq A \right) \leq \tr{\Var(X)A^{-2}}.
%     \end{equation}
% \end{lemma}
% For a detailed proof, see Theorem 14 from \cite{ahlswede2001strong}. Briefly, the proof of Lemma \ref{lemma:matrix_chebyshev} follows by observing that $(X - \E[X])^2 \preceq A^2 \implies \abslr{X - \E[X]} \preceq A$, which is equivalent to $\abslr{X - \E[X]} \not\preceq A \implies (X - \E[X])^2 \not\preceq A^2$. Now, since $(X - \E[X])^2 \succeq 0$, we can use the Markov's inequality for matrices from Lemma \ref{lemma:matrix_markov}, where $X$ in \eqref{eqn:matrix_markov} is replaced with $(X - \E[X])^2$ and $A$ with $A^2$.

On several occasions in the proof of Theorem \ref{thm:vector_case}, we need to bound the Frobenius norm of a random matrix. In this regard, the following lemma defines a variance statistic $\text{var}(.)$ and shows how it exploits independence.
\begin{lemma} \label{lemma:var_statistic}
    Let $\{X_i\}_{i \in [n]}$ be a collection of independent matrices. Define $\text{var}(X_i) \triangleq \E[\normlr{X_i - \E[X_i]}_F^2]$. We have $$\text{var}\left(\sum_{i \in [n]} X_i \right) = \sum_{i \in [n]} \text{var}(X_i).$$
\end{lemma}
\begin{proof}
    Using the definition of $\text{var}(\cdot)$, we have 
    \begin{align*}
       \text{var}\left(\sum_{i \in [n]} X_i \right) & = \E\left[\normlr{\sum_{i \in [n]} X_i - \E[X_i]}_F^2\right] \\
       & \overset{(a)}{=}  \E\left[\tr{\left(\sum_{i \in [n]} X_i - \E[X_i] \right) \left(\sum_{i \in [n]} X_i - \E[X_i] \right)^{\top}}\right] \\
       & \overset{(b)}{=} \tr{\E\left[ \left(\sum_{i \in [n]} X_i - \E[X_i] \right) \left(\sum_{i \in [n]} X_i - \E[X_i] \right)^{\top} \right]} \\
       & = \tr{\E \left[ \sum_{i \in [n]} (X_i - \E[X_i])(X_i - \E[X_i])^{\top} \right]}  \\
       & \quad + \tr{\E \left[ \sum_{i \neq j, i, j\in [n]} (X_i - \E[X_i])(X_j - \E[X_j])^{\top} \right]} \\
       & \overset{(c)}{=} \tr{\E \left[ \sum_{i \in [n]} (X_i - \E[X_i])(X_i - \E[X_i])^{\top} \right]} \\
       & \overset{(d)}{=}\sum_{i \in [n]} \tr{\E\left[(X_i - \E[X_i])(X_i - \E[X_i])^{\top}\right]} \\
       & \overset{(e)}{=} \sum_{i \in [n]} \E\left[\tr{(X_i - \E[X_i])(X_i - \E[X_i])^{\top}}\right] \\
       & \overset{(f)}{=} \sum_{i \in [n]} \E\left[\normlr{X_i - \E[X_i]}_F^2\right] \\
       & = \sum_{i \in [n]} \text{var}(X_i).
    \end{align*}
    In the above steps, (a) and (f) follow from the definition of the Frobenius norm, (b), (d), and (e) from the linearity of expectation and trace operators, and (c) follows due to independence as the second term vanishes as a result of taking the expectation.
\end{proof}

\begin{lemma} (\textbf{Property of the geometric median})
\label{lemma:geometric_median}
Let $A_1, \ldots, A_K \in \mathbb{R}^{d \times d}$ and let $$A_* = \texttt{Med}(A_1, \ldots, A_K) := \argmin_{\theta \in \mathbb{R}^{d \times d}} \sum_{j \in [K]} \Vert \theta - A_j \Vert_F$$ be their geometric median. Fix $\alpha \in (0, 0.5)$ and assume $A \in \mathbb{R}^{d \times d}$ is such that $\normlr{A_* - A}_F > C_\alpha r$, where $C_\alpha = (1 - \alpha) \sqrt{1/(1 - 2\alpha)}$ and $r > 0$. Then,  there exists a subset $J \subseteq [K]$ of cardinality $\abs{J} > \alpha K$ such that for all $j \in J$, $\normlr{A_j - A}_F > r$.
\end{lemma}
See Lemma 2.1 from \cite{minsker2015geometric} for the proof of the above lemma. Lemma \ref{lemma:geometric_median} helps us in the boosting step where we improve the bounds by fusing the OLS estimators obtained from different buckets using the geometric median operator.

\section{Proof of Theorem \ref{thm:scalar_case}: The scalar case}
\label{app:scalar_proof}
In this section, we prove the results pertaining to the scalar case as presented in Theorem \ref{thm:scalar_case}. For readers' convenience, we restate the lemmas from the main text before providing their proofs. We begin with Lemma \ref{lemma:scalar_num}, which provides an upper-bound for the numerator of the error term in \eqref{eqn:ols_scalar}.

\begin{lemma} (\textbf{Scalar numerator upper bound})  Fix a bucket $j \in [K]$. With probability at least $1 - p/2$, the following holds:
$$\abslr{\sum_{i \in \mc{B}_j} x_T^{(i)}w_T^{(i)}} \leq c\sigma^2\sqrt{g_T M/p}.$$ 
\end{lemma}
\begin{proof}
To prove the above result, one can start by noting that due to the i.i.d nature of the trajectories, $\Var(\sum_{i \in \mc{B}_j} x_T^{(i)}w_T^{(i)}) = M \Var(x_T^{(1)}w_T^{(1)})$, where $M = \abs{\mc{B}_j}$, and $\Var(Z)$ is used to represent the variance of a real-valued random variable $Z$. For clarity of notation, we drop the superscript in the rest of the proof. Observe that for each individual term,
\begin{equation}
    \Var(x_Tw_T) = \mathbb{E}[x_T^2] \mathbb{E}[w_T^2] = \sigma^2 g_T \times \sigma^2 = \sigma^4 g_T,
\end{equation} 
where, we exploited the fact that $x_T$ and $w_T$ are independent, $\mathbb{E}[x_T] = \mathbb{E}[w_T]=0$, and, defining $n_t \triangleq w_{T - (t+1)}$, we have
\begin{equation*}
\begin{aligned}
    \mathbb{E}[x_T^2] &= \mathbb{E}\left[\left(\sum_{t = 0}^{T-1} a^tn_t\right)^2\right] \\
    &= \mathbb{E}\left[\sum_{t = 0}^{T-1} a^{2t}n_t^2\right] \\
    &= \sigma^2 g_T.
\end{aligned}
\end{equation*}
The rest follows from a straightforward application of Chebyshev's inequality. For any $t >0$, we have
\begin{equation*}
    \mathbb{P}\left(\abslr{\sum_{i \in \mc{B}_j} x_T^{(i)}w_T^{(i)}} \geq t \right) \leq \frac{\sigma^4 g_T M}{t^2}.
\end{equation*}
Setting the R.H.S. of the above inequality to $p/2$, and solving for $t$ completes the proof.
\end{proof}

Similarly, Lemma \ref{lemma:scalar_den}, whose proof is included in the main text, provides a lower-bound for the denominator of the error term in \eqref{eqn:ols_scalar}. Combining the results from Lemma \ref{lemma:scalar_num} and Lemma \ref{lemma:scalar_den} using an union bound, we have that when $M \geq (c_1/p) (\tilde{\sigma}^4/\sigma^4)$, the following holds with probability at least $1 - p$:
\begin{equation}
    \abs{\hat{a}_j - a} \leq c_2\sqrt{\frac{1}{p M g_T}}.
    \label{eqn:scalar_bound_pre_boost_appendix}
\end{equation}

Next, the following result pertains to the boosting step of our algorithm.
\begin{lemma} \label{lemma:scalar_boosting} (\textbf{Scalar boosting}) In \eqref{eqn:scalar_bound_pre_boost_appendix}, fix $p = 1/4$. Let $\hat{a} = \texttt{Med}(\hat{a}_1, \ldots, \hat{a}_K)$, where $\hat{a}_1, \ldots, \hat{a}_K$ are $K$ independent OLS estimators corresponding to different buckets, each satisfying \eqref{eqn:scalar_bound_pre_boost_appendix}. Given $\delta > 0$, when $K = \ceil{c_1 \log(1/\delta)}$, the following holds with probability at least $1 - \delta:$
$$\abs{\hat{a} - a} \leq C \sqrt{\frac{\log(1/\delta)}{N g_T}}.$$
\end{lemma}
\begin{proof}
    Let $\varepsilon = c_2({p M g_T})^{-1/2}.$ By the property of the median, observe that the event $\{ \abslr{\hat{a}-a} > \varepsilon\}$ implies $\{ \sum_{j \in [K]} Y_j \geq K/2 \}$, where $Y_j$ is an indicator random variable of the event  $\{ \abslr{\hat{a}_j-a} > \varepsilon\}$. This implies
    \begin{align*}
        \mathbb{P}\left( \abslr{\hat{a}-a} > \varepsilon \right) \leq \mathbb{P} \left( \sum_{j \in [K]} Y_j \geq K/2 \right) \overset{(a)}{=} \mathbb{P} \left( \frac{1}{K} \sum_{j \in [K]} (Y_j - \E[Y_1]) \geq \frac{1}{2}  - \E[Y_1] \right).
    \end{align*}
    In (a), we subtracted the common expectation since the trajectories are identically distributed. Furthermore, since the collections of trajectories from different buckets are disjoint, each of the $Y_j$'s are i.i.d. random variables in $\{0,1\}$. This enables us to use Hoeffding's inequality to infer that
    \begin{align*}
        \mathbb{P} \left( \frac{1}{K} \sum_{j \in [K]} (Y_j - \E[Y_1]) \geq \frac{1}{2}  - \E[Y_1] \right) &\leq \exp{\left(-2 K (1/2 - \E[Y_1])^2\right)} \\
        & \overset{(b)}{\leq} \exp{\left(-2 K (1/2 - 1/4)^2\right)}.
    \end{align*}
    In the above steps, (b) follows due to $\E[Y_1] = \mathbb{P}(\abslr{\hat{a}_1 - a} > \varepsilon) \leq p = 1/4$, where we used~\eqref{eqn:scalar_bound_pre_boost_appendix}.  Based on the above, we have
    $$ \mathbb{P}\left( | \hat{a}-a| > \varepsilon\right) \leq \exp(-K/8) \leq \delta,$$
    when $K=\ceil{8 \log(1/\delta)}$. Using this expression for $K$ in $\varepsilon= c_2\sqrt{\frac{4K}{N g_T}}$, where we used $M=N/K$, completes the proof of  Lemma \ref{lemma:scalar_boosting}.
\end{proof}
Finally, the proof for Theorem \ref{thm:scalar_case} follows directly from the result of Lemma \ref{lemma:scalar_boosting}, and noting that $M \geq c_2 (\tilde{\sigma}^2/\sigma^2)$ to ensure that the OLS estimator corresponding to each bucket satisfies \eqref{eqn:scalar_bound_pre_boost_appendix}. 

%%%%%%%%%%%%%%%%%%%%%%%%%%%%%%%%%%%%%%%
\newpage
\section{Proof of Theorem \ref{thm:vector_case}: The vector case}
\label{app:vector_proof}
In this section, we prove the results pertaining to the vector case as presented in Theorem \ref{thm:vector_case}. Analogous to the scalar case, we first bound the terms that constitute the error term of the OLS estimator for each bucket, and then combine them by applying an union bound. Finally, we fuse the OLS estimators corresponding to different buckets using the geometric median. The OLS estimator $\hat{A}_j$ for the $j$th bucket can be expressed as
\begin{equation*}
    \hat{A}_j = A + \sum_{i \in \mc{B}_j} w_T^{(i)}\big(x_T^{(i)}\big)^{\top}\left(\sum_{i \in \mc{B}_j} x_T^{(i)}\big(x_T^{(i)}\big)^{\top}\right)^{-1}.
    \label{eqn:ols_vector}
\end{equation*}
To simplify our analysis, we whiten the vector $x_T^{(i)}$ and define $z_T^{(i)} = \Sigma_x^{-1/2}x_T^{(i)}$, where $\Sigma_x = \E[x_T^{(i)}\big(x_T^{(i)}\big)^{\top}]$. For each $j \in [K]$, the error of the OLS estimator can be bounded as follows: 
\begin{equation*}
    \begin{aligned}
        \norm{\hat{A}_j - A} &= \normlr{\sum_{i \in \mc{B}_j} w_T^{(i)}\big(z_T^{(i)}\big)^{\top}\Sigma_x^{1/2}\left(\Sigma_x^{1/2} \bigg(\sum_{i \in \mc{B}_j} z_T^{(i)}\big(z_T^{(i)}\big)^{\top}\bigg)\Sigma_x^{1/2}\right)^{-1}} \\
        &\leq \normlr{\Sigma_x^{-1/2}} \frac{\normlr{\sum_{i \in \mc{B}_j} w_T^{(i)}\big(z_T^{(i)}\big)^{\top}}}{\lambda_{\min}{\left(\sum_{i \in \mc{B}_j} z_T^{(i)}\big(z_T^{(i)}\big)^{\top}\right)}},
    \end{aligned}
\end{equation*}
where the inequality follows from submultiplicativity of spectral norm. Based on the above expression, it suffices to individually bound $\normlr{\sum_{i \in \mc{B}_j} w_T^{(i)}\big(z_T^{(i)}\big)^{\top}}$ and $\lambda_{\min}{\left(\sum_{i \in \mc{B}_j} z_T^{(i)}\big(z_T^{(i)}\big)^{\top}\right)}$. The following lemmas, which are the restated versions of Lemma \ref{lemma:vector_num} and Lemma \ref{lemma:vector_den} in the main text, provide key results in this regard.

\begin{lemma} (\textbf{Vector numerator upper bound}) Fix a bucket $j \in [K]$. With probability at least $1 - p/2$, the following holds:
$$\normlr{\sum_{i \in \mc{B}_j} w_T^{(i)}\big(z_T^{(i)}\big)^{\top}} \leq c_1 d \sqrt{\sigma^2 M/p}.$$
\end{lemma}
\begin{proof}
    We begin by noting that $\normlr{\sum_{i \in \mc{B}_j} w_T^{(i)}\big(z_T^{(i)}\big)^{\top}}^2 \leq \normlr{\sum_{i \in \mc{B}_j} w_T^{(i)}\big(z_T^{(i)}\big)^{\top}}_F^2$. Therefore, it suffices to bound the term $\normlr{\sum_{i \in \mc{B}_j} w_T^{(i)}\big(z_T^{(i)}\big)^{\top}}_F^2$. Using Markov's inequality, we have for any $t > 0$
    \begin{align}
        \mathbb{P}\left(\normlr{\sum_{i \in \mc{B}_j} w_T^{(i)}\big(z_T^{(i)}\big)^{\top}}_F^2 \geq t^2\right) & \leq 
        \frac{\mathbb{E}\left[ \normlr{\sum_{i \in \mc{B}_j} w_T^{(i)}\big(z_T^{(i)}\big)^{\top}}_F^2\right]}{t^2} \label{eqn:vec_num_1}\\
        & \overset{(a)}{=} \frac{\sum_{i \in \mc{B}_j} \mathbb{E}\left[ \normlr{w_T^{(i)}\big(z_T^{(i)}\big)^{\top}}_F^2\right]}{t^2} \nonumber\\
        & \overset{(b)}{=} \frac{M \mathbb{E}\left[ \normlr{w_T^{(1)}\big(z_T^{(1)}\big)^{\top}}_F^2\right]}{t^2} \nonumber\\
        & \overset{(c)}{=} \frac{M \mathbb{E}\left[ \tr{w_T\big(z_T\big)^{\top}z_T\big(w_T\big)^{\top}}\right]}{t^2} \nonumber\\
        & \overset{(d)}{=} \frac{M \mathbb{E}\left[ \tr{\big(z_T\big)^{\top}z_T w_T\big(w_T\big)^{\top}}\right]}{t^2} \nonumber\\
        & \overset{(e)}{=} \frac{M \tr{\mathbb{E}\left[ \norm{z_T}^2 \right] \E\left[w_T\big(w_T\big)^{\top} \right]}}{t^2} \nonumber\\
        & \overset{(f)}{=} \frac{M d^2 \sigma^2}{t^2}. \label{eqn:frob_norm_bnd}
    \end{align}
    In the above steps, (a) follows from Lemma \ref{lemma:var_statistic}; (b) since the trajectories are identically distributed; in (c), we use the definition of the Frobenius norm and drop the superscript notation for clarity; (d) follows from moving the scalar term $\big(z_T\big)^{\top}z_T$ within the trace; (e) uses linearity of expectation and the independence of $z_T$ and $w_T$; and finally, (f) follows from our assumption on the noise process $w_T$ and the fact that $z_T$ is whitened. Setting the R.H.S equal to $p/2$ and solving for $t$, we have with probability at least $1- p/2$, 
    \begin{equation}
        \normlr{\sum_{i \in \mc{B}_j} w_T^{(i)}\big(z_T^{(i)}\big)^{\top}} \leq \normlr{\sum_{i \in \mc{B}_j} w_T^{(i)}\big(z_T^{(i)}\big)^{\top}}_F \leq  c_1 d \sqrt{\sigma^2 M/p}.
    \end{equation}
    This completes the proof. 

    \textbf{Alternate Proof.} As discussed in the main text, one could obtain the same bound by using the matrix version of Markov's inequality from Lemma \ref{lemma:matrix_markov} as follows:
    \begin{align*}
        & \mathbb{P}\left( \left(\sum_{i \in \mc{B}_j} w_T^{(i)}\big(z_T^{(i)}\big)^{\top}\right) \left(\sum_{i \in \mc{B}_j} w_T^{(i)}\big(z_T^{(i)}\big)^{\top}\right)^{\top} \not \preceq t^2 I \right) \\
        \quad & \overset{(a)}{\leq} \frac{\tr{\E\left[ \left(\sum_{i \in \mc{B}_j} w_T^{(i)}\big(z_T^{(i)}\big)^{\top}\right) \left(\sum_{i \in \mc{B}_j} w_T^{(i)}\big(z_T^{(i)}\big)^{\top}\right)^{\top}\right]}}{t^2} \\
        & \overset{(b)}{=} \frac{\E\left[ \tr{\left(\sum_{i \in \mc{B}_j} w_T^{(i)}\big(z_T^{(i)}\big)^{\top}\right) \left(\sum_{i \in \mc{B}_j} w_T^{(i)}\big(z_T^{(i)}\big)^{\top}\right)^{\top}}\right]}{t^2} \\
        & \overset{(c)}{=} \frac{\mathbb{E}\left[ \normlr{\sum_{i \in \mc{B}_j} w_T^{(i)}\big(z_T^{(i)}\big)^{\top}}_F^2\right]}{t^2} \\
        & \overset{(d)}{=} \frac{M d^2 \sigma^2}{t^2}.
    \end{align*}
    In the above steps, (a) follows from \eqref{eqn:matrix_markov}, (b) by commuting the trace and expectation operators, (c) by using the definition of the Frobenius norm, and (d) from the derivations leading up to \eqref{eqn:frob_norm_bnd}. Notice that we have obtained the same upper-bound on the failure probability as in our prior analysis involving the variance statistic. Equating the failure probability to $p/2$, we have with probability at least $1 - p/2$, 
    \begin{align*}
        \left(\sum_{i \in \mc{B}_j} w_T^{(i)}\big(z_T^{(i)}\big)^{\top}\right) \left(\sum_{i \in \mc{B}_j} w_T^{(i)}\big(z_T^{(i)}\big)^{\top}\right)^{\top} &\preceq \left(c_1d \sqrt{\sigma^2 M/p}\right)^2 I \\
        \implies \normlr{\sum_{i \in \mc{B}_j} w_T^{(i)}\big(z_T^{(i)}\big)^{\top}} & \leq c_1 d \sqrt{\sigma^2 M/p}.
    \end{align*}
    
\end{proof}

Before proving Lemma \ref{lemma:vector_den}, which provides a lower-bound for $\lambda_{\min}{\left(\sum_{i \in \mc{B}_j} z_T^{(i)}\big(z_T^{(i)}\big)^{\top}\right)}$, we prove Lemma \ref{lemma:vec_den_biquadratic_terms} which is a key ingredient in the proof of Lemma \ref{lemma:vector_den}.

\begin{lemma} Define $n_t \triangleq w_{T-(t+1)}.$ Given the system in \eqref{eqn:sys_model} and the noise assumptions in \eqref{eqn:noise_model}, we have
    \begin{equation}
    \resizebox{0.8\hsize}{!}{$
    \begin{aligned}
\E\left[(x_Tx_T^{\top})^2\right] &= \sum_{t=0}^{T-1} \E\left[(A^tn_tn_t^{\top}(A^t)^{\top})^2\right] 
         + 2\sum_{s\neq t=0}^{T-1} \E\left[ A^tn_tn_t^{\top}(A^t)^{\top} A^sn_sn_s^{\top}(A^s)^{\top}\right]\\ 
         &+ \sum_{s\neq t=0}^{T-1} \E\left[ A^tn_tn_s^{\top}(A^s)^{\top} A^sn_sn_t^{\top}(A^t)^{\top}\right]. 
         \nonumber 
\end{aligned}$}
    \end{equation}
\end{lemma}

\begin{proof}
    We have $x_T = \sum_{t = 0}^{T-1} A^tn_t$. In what follows, we simplify the summation notation by dropping the ranges, which vary from $0$ to $T-1$, and preserving only the relation between the indices for clarity. $\E\left[(x_Tx_T^{\top})^2\right]$ can be expressed as
    \begin{align*}
        \E\left[(x_Tx_T^{\top})^2\right] &= \E\left[ \left( \sum_{k}A^kn_k \sum_{l}n_l^{\top}(A^l)^{\top} \right) \left( \sum_{s}A^sn_s \sum_{t}n_t^{\top}(A^t)^{\top} \right)\right] \\
        &= \E \left( \underbrace{\sum_k A^kn_kn_k^{\top}(A^k)^{\top}}_{T_1} + \underbrace{\sum_{k\neq l} A^kn_kn_l^{\top}(A^l)^{\top}}_{T_2} \right) \\
        & \times \left( \underbrace{\sum_s A^sn_sn_s^{\top}(A^s)^{\top}}_{T_3} + \underbrace{\sum_{s\neq t} A^sn_sn_t^{\top}(A^t)^{\top}}_{T_4} \right).
    \end{align*}
    Next, we expand the above expression, and analyze the result of taking the expectation on each of the four terms that arise in the above product, starting with the term $T_1 \times T_3.$ We have 
    \begin{align*}
        \E[T_1 \times T_3] = \sum_{k} \E\left[(A^kn_kn_k^{\top}(A^k)^{\top})^2\right] + \sum_{k\neq s} \E\left[ A^kn_kn_k^{\top}(A^k)^{\top} A^sn_sn_s^{\top}(A^s)^{\top}\right]. 
    \end{align*}
    In the following, as a result of the noise process being i.i.d. with zero mean, we show that the term $T_1 \times T_4$ does not survive the expectation.
    \begin{align*}
        \E[T_1 \times T_4] & = \sum_{s \neq t} \sum_k \E\left[ A^kn_kn_k^{\top}(A^k)^{\top} A^sn_sn_t^{\top}(A^t)^{\top}\right] \\
        & \overset{(a)}{=}  \sum_k \sum_{s \neq t = k} \E\left[ A^kn_kn_k^{\top}(A^k)^{\top} A^sn_sn_k^{\top}(A^k)^{\top}\right] \\
        & + \sum_k \sum_{t \neq s = k} \E\left[ A^kn_kn_k^{\top}(A^k)^{\top} A^kn_kn_t^{\top}(A^t)^{\top}\right] \\
        & + \sum_k \sum_{t \neq s \neq k} \E\left[ A^kn_kn_k^{\top}(A^k)^{\top} A^sn_sn_t^{\top}(A^t)^{\top}\right] = 0.
    \end{align*}
    Note that all the three expectations in (a) evaluate to a $0$ matrix. For example, in the first term of (a), the embedded scalar term $n_k^{\top}(A^k)^{\top} A^sn_s$ can be moved to the right end, yielding: 
    $$ \E\left[ A^kn_kn_k^{\top}(A^k)^{\top} A^sn_sn_k^{\top}(A^k)^{\top}\right] = \E\left[ A^kn_k n_k^{\top}(A^k)^{\top} n_k^{\top}(A^k)^{\top}\right]  \E\left[ A^sn_s \right] =0,$$
where we used the fact that $n_k$ and $n_s$ are independent, and $\E[n_s]=0.$ A similar argument can be used to conclude that the second and third terms in (a) also evaluate to $0$. 

    Similarly, since $\E[T_1 \times T_4] = \E[T_2 \times T_3]$, we conclude that $\E[T_1 \times T_4] = 0$. Next, we analyze the term $\E[T_2 \times T_4]$ in the following:
    \begin{align*}
        \E[T_2 \times T_4] & = \sum_{k \neq l} \sum_{s \neq t} \E\left[A^kn_kn_l^{\top}(A^l)^{\top} A^sn_sn_t^{\top}(A^t)^{\top}\right] \\
        & \overset{(a)}{=} \sum_{k\neq l} \E\left[ A^kn_kn_k^{\top}(A^k)^{\top} A^ln_ln_l^{\top}(A^l)^{\top}\right] \\
        & + \sum_{k\neq l} \E\left[ A^kn_kn_l^{\top}(A^l)^{\top} A^ln_ln_k^{\top}(A^k)^{\top}\right] \\
        & \overset{(b)}{+} \sum_{l \neq k = t \neq s \neq l} \E\left[ A^kn_kn_l^{\top}(A^l)^{\top} A^sn_sn_k^{\top}(A^k)^{\top}\right] \\
        & + \sum_{l \neq k = s \neq t \neq l} \E\left[ A^kn_kn_l^{\top}(A^l)^{\top} A^kn_kn_t^{\top}(A^t)^{\top}\right] \\
        & + \sum_{k \neq l = t \neq s \neq k} \E\left[ A^kn_kn_l^{\top}(A^l)^{\top} A^sn_sn_l^{\top}(A^l)^{\top}\right] \\
        & + \sum_{k \neq l = s \neq t \neq k} \E\left[ A^kn_kn_l^{\top}(A^l)^{\top} A^ln_ln_t^{\top}(A^t)^{\top}\right] \\
        & + \sum_{k \neq l \neq s \neq t \neq k, s \neq k, t \neq l} \E\left[ A^kn_kn_l^{\top}(A^l)^{\top} A^sn_sn_t^{\top}(A^t)^{\top}\right].
    \end{align*}
    In the above display, it is not hard to verify that except the first and second terms of (a), all the remaining terms starting from (b) evaluate to zero. To see this, notice that the embedded scalar terms of the form $n^{\top}A^{\top} An$ can be flipped or moved within the expectation. Then, using the i.i.d assumption and independence, one can verify the aforementioned claim. For completeness, we show how this can be done for the first and second terms starting from (b); the rest follow similarly. For the first term, we have
$$ \E\left[ A^kn_k n_l^{\top}(A^l)^{\top} A^sn_sn_k^{\top}(A^k)^{\top}\right] =  \E\left[ A^kn_k n_k^{\top} (A^k)^{\top}\right] \E\left[ n_l^{\top}(A^l)^{\top} \right] \E\left[ A^sn_s  \right] =0,$$
where we used independence of $n_k, n_l,$ and $n_s$, and $\E\left[ n_s \right]=0.$ Similarly, for the second term,  
\begin{align*}
    \E\left[ A^kn_kn_l^{\top}(A^l)^{\top} A^kn_kn_t^{\top}(A^t)^{\top}\right] &= \E\left[ A^kn_k n_k^{\top}(A^k)^{\top} A^l n_l n_t^{\top}(A^t)^{\top}\right] \\
    &= \E\left[ A^kn_k n_k^{\top} (A^k)^{\top}\right] \E\left[ A^l n_l  \right] \E\left[ n_t^{\top}(A^t)^{\top} \right],
\end{align*}
which evaluates to $0$, since $\E \left [ n_l \right] =0.$ Proceeding similarly, we conclude that the terms starting from (b) do not survive the expectation. Therefore, we have
    \begin{align*}
        \E[T_2 \times T_4] & = \sum_{k=s \neq l=t} \E\left[ A^kn_kn_k^{\top}(A^k)^{\top} A^ln_ln_l^{\top}(A^l)^{\top}\right] \\
        & + \sum_{k=t \neq l=s} \E\left[ A^kn_kn_l^{\top}(A^l)^{\top} A^ln_ln_k^{\top}(A^k)^{\top}\right].
    \end{align*}
    Finally, compiling the results of $\E[T_1 \times T_3]$ and $\E[T_2 \times T_4]$, we obtain the desired claim in the lemma.
\end{proof}

Equipped with the Lemma \ref{lemma:vec_den_biquadratic_terms}, we are now ready to prove the lower-bound for $\lambda_{\min}{\left(\sum_{i \in \mc{B}_j} z_T^{(i)}\big(z_T^{(i)}\big)^{\top}\right)}$ as stated in Lemma \ref{lemma:vector_den}.

\begin{lemma} (\textbf{Vector denominator lower bound}) For each bucket $j$, with probability at least $1 - p/2$, $$\lambda_{\min}\left(\sum_{i \in \mc{B}_j} z_T^{(i)}\big(z_T^{(i)}\big)^{\top}\right) \geq M/2, \hspace{2mm} \textrm{provided} \hspace{2mm} M \geq c (d^2/p) C_A C_w.$$
\end{lemma}
\begin{proof}
    Since $\E\left[ z_T^{(i)}\big(z_T^{(i)}\big)^{\top} \right] = I \ \forall i \in \mc{B}_j$, we have 
    \begin{align}
        \sum_{i \in \mc{B}_j}  z_T^{(i)}\big(z_T^{(i)}\big)^{\top} &= \sum_{i \in \mc{B}_j} \E\left[ z_T^{(i)}\big(z_T^{(i)}\big)^{\top}\right] + \sum_{i \in \mc{B}_j} \left( z_T^{(i)}\big(z_T^{(i)}\big)^{\top} - \E\left[ z_T^{(i)}\big(z_T^{(i)}\big)^{\top}\right] \right) \nonumber \\
        &= M I + \sum_{i \in \mc{B}_j} \left( z_T^{(i)}\big(z_T^{(i)}\big)^{\top} - I \right) \label{eqn:vec_den_good_error}.
    \end{align}
    Based on the above, we have $\lambda_{\min}{\left(\sum_{i \in \mc{B}_j} z_T^{(i)}\big(z_T^{(i)}\big)^{\top}\right)} = M + \lambda_{\min}{\left(\sum_{i \in \mc{B}_j} \left( z_T^{(i)}\big(z_T^{(i)}\big)^{\top} - I \right)\right)}.$
    Therefore, it suffices to obtain a lower-bound for $\lambda_{\min}{\left(\sum_{i \in \mc{B}_j} \left( z_T^{(i)}\big(z_T^{(i)}\big)^{\top} - I \right)\right)}$. To do so, similar to the procedure followed in the proof of Lemma \ref{lemma:vector_num}, one could either use the variance statistic from Lemma \ref{lemma:var_statistic} or the matrix version of Markov's inequality presented in Lemma \ref{lemma:matrix_markov}. Both the methods lead to a similar analysis, and in the following, we proceed with the latter approach. We define $X \triangleq \sum_{i \in \mc{B}_j} \left( z_T^{(i)}\big(z_T^{(i)}\big)^{\top} - I \right)$ for brevity. We have
    \begin{align}
        \mathbb{P}\left( XX^{\top} \not \preceq t^2 I \right) & \leq \frac{\tr{\E[XX^{\top}]}}{t^2} \nonumber \\
        & \overset{(a)}{=} \frac{\E[\norm{X}_F^2]}{t^2} \nonumber \\
        & \overset{(b)}{=} \frac{M \tr{\E\left[\left(z_T^{(1)}\big(z_T^{(1)}\big)^{\top} - I\right)^2\right]}}{t^2}, \label{eqn:vec_den_cheby}
    \end{align}
    where (a) follows by commuting trace and expectation operators and using the definition of Frobenius norm, and (b) follows from Lemma \ref{lemma:var_statistic} due to the i.i.d. nature of the trajectories. In the following, we drop the superscript notation for clarity. Furthermore, in the analysis that follows, we exploit the linearity of trace and expectation at several points as $\tr{\E[x]} = \E[\tr{x}]$ where $x$ is any random matrix in $\mathbb{R}^{d \times d}$. We also use the cyclic property of trace at several steps. Now, we derive an upper-bound for $\tr{\E\left[\left(z_T z_T^{\top} - I\right)^2\right]}$ as follows: 
    \begin{align}
        \tr{\E\left[\left(z_T z_T^{\top} - I\right)^2\right]} & \overset{(a)}{=} \tr{\E\left[\left(z_T z_T^{\top}\right)^2\right]} - d \nonumber\\
        & \leq \tr{\E\left[\left(z_T z_T^{\top}\right)^2\right]} \nonumber\\
        & = \E\left[\tr{\left(z_T z_T^{\top}\right)^2}\right] \nonumber\\
        & \overset{\eqref{eqn:tr_a_sq}}{=} \E\left[\left(\tr{z_T z_T^{\top}}\right)^2\right] \nonumber\\
        & = \E\left[\left(\tr{\Sigma_x^{-1/2}x_T x_T^{\top}\Sigma_x^{-1/2}}\right)^2\right] = \E\left[\left(\tr{\Sigma_x^{-1}x_T x_T^{\top}}\right)^2\right] \nonumber\\
        &  \overset{\eqref{eqn:tr_ab_tr_a_lambda_b}}{=} \E\left[ \left( \lambda_{\max}{(\Sigma_x^{-1})}\tr{x_T x_T^{\top}}\right)^2\right] \nonumber\\
        & = \frac{\E\left[ \left(\tr{x_T x_T^{\top}}\right)^2\right]}{\left(\lambda_{\min}{(\Sigma_x)}\right)^2} \nonumber\\
        & \overset{\eqref{eqn:tr_a_sq}}{=} \frac{\E\left[\tr{\left(x_T x_T^{\top}\right)^2}\right]}{\left(\lambda_{\min}{(\Sigma_x)}\right)^2} \nonumber\\
        & = \frac{\tr{\E\left[\left(x_T x_T^{\top}\right)^2\right]}}{\left(\lambda_{\min}{(\Sigma_x)}\right)^2}.\label{eqn:vec_den_trace_bound}
    \end{align}
    In the above steps, (a) follows since $\E\left[z_T z_T^{\top}\right] = I$. Next, based on the last step, we obtain an upper-bound on $\tr{\E\left[\left(x_T x_T^{\top}\right)^2\right]}$. From Lemma \ref{lemma:vec_den_biquadratic_terms}, we have
     \begin{align}
    \E\left[(x_Tx_T^{\top})^2\right] &= \underbrace{\sum_{t=0}^{T-1} \E\left[(A^tn_tn_t^{\top}(A^t)^{\top})^2\right]}_{T_1}
         + 2\underbrace{\sum_{s\neq t=0}^{T-1} \E\left[ A^tn_tn_t^{\top}(A^t)^{\top} A^sn_sn_s^{\top}(A^s)^{\top}\right]}_{T_2} \nonumber\\ 
         &+ \underbrace{\sum_{s\neq t=0}^{T-1} \E\left[ A^tn_tn_s^{\top}(A^s)^{\top} A^sn_sn_t^{\top}(A^t)^{\top}\right]}_{T_3}. 
         \label{eqn:vec_den_t1_t2_t3} 
    \end{align}
In the following, we bound the the trace of terms in the above equation starting with the term $T_1$. we have
\begin{align}
    \tr{T_1} &= \sum_{t=0}^{T-1} \E\left[\tr{(A^tn_tn_t^{\top}(A^t)^{\top})^2}\right] \nonumber\\
    &\overset{\eqref{eqn:tr_a_sq}}{=} \sum_{t=0}^{T-1} \E\left[\left(\tr{A^tn_tn_t^{\top}(A^t)^{\top}}\right)^2\right] \nonumber\\
    & \overset{\eqref{eqn:tr_ab_tr_a_lambda_b}}{\leq} \sum_{t=0}^{T-1} \E\left[\left(\tr{n_tn_t^{\top}}\right)^2\big(\lambda_{\max}(A^t(A^t)^{\top})\big)^2\right] \nonumber\\
    & = \sum_{t=0}^{T-1} \E\left[\left(\tr{n_t^{\top}n_t}\right)^2\big(\lambda_{\max}(A^t(A^t)^{\top})\big)^2\right] \nonumber\\
    &\overset{(a)}{=} \E\left[\norm{n_1}^4\right]\sum_{t=0}^{T-1} \norm{A^t}^4 \nonumber\\
    &\overset{(b)}{\leq} d^2 \tilde{\sigma}^4 \sum_{t=0}^{T-1} \norm{A^t}^4 \label{eqn:t1bound}.
\end{align}
In the above steps, (a) follows as the noise sequence is identically distributed for all $t$, and (b) follows by observing that 
\begin{align*}
    \E[\norm{n_1}^4] = \E\left[ \left(\sum_{i = 1}^d n_1(i)^2 \right)^2 \right] &= \E\left[\sum_{i = 1}^d n_1(i)^4 \right]+\E\left[\sum_{i \neq j = 1}^d n_1(i)^2 n_1(j)^2 \right] \\
    &\overset{(c)}{=} d\tilde{\sigma}^4 + (d^2 - d)\sigma^4 \\
    &\overset{(d)}{\leq} d^2 \tilde{\sigma}^4,
\end{align*}
where (c) follows based on the noise assumptions \eqref{eqn:noise_model}, and (d) as $\tilde{\sigma}^4 \geq \sigma^4$ due to Jensen's inequality. Similarly, we bound the trace of the term $T_2$ from \eqref{eqn:vec_den_t1_t2_t3} in the following.

\begin{align}
    \tr{T_2} &= \tr{\sum_{s\neq t=0}^{T-1} \E\left[ A^tn_tn_t^{\top}(A^t)^{\top} A^sn_sn_s^{\top}(A^s)^{\top}\right]} \nonumber\\
    &\overset{(a)}{=} \sum_{s\neq t=0}^{T-1} \tr{A^t\E[n_tn_t^{\top}](A^t)^{\top} A^s\E[n_sn_s^{\top}](A^s)^{\top}} \nonumber\\
    &\overset{(b)}{=} \sigma^4 \sum_{s\neq t=0}^{T-1}  \tr{A^t(A^t)^{\top} A^s(A^s)^{\top}} \nonumber\\
    &\overset{\eqref{eqn:tr_ab_eq_tra_trb}}{\leq} \sigma^4 \sum_{s\neq t=0}^{T-1} \tr{A^t(A^t)^{\top}} \tr{A^s(A^s)^{\top}} \nonumber\\
    &\overset{\eqref{eqn:tr_a_d_lambda}}{\leq} d^2 \sigma^4 \sum_{s\neq t=0}^{T-1} \lambda_{\max}(A^t(A^t)^{\top}) \lambda_{\max}(A^s(A^s)^{\top}) \nonumber\\
    & \overset{(c)}{\leq} d^2 \tilde{\sigma}^4 \sum_{s\neq t=0}^{T-1} \norm{A^t}^2 \norm{A^s}^2. \label{eqn:t2bound}
\end{align}
In the above steps, (a) follows from independence of $n_t$ and $n_s$, (b) from the noise assumptions, and (c) as $\tilde{\sigma}^4 \geq \sigma^4$ due to Jensen's inequality. Next, we bound the trace of the term $T_3$ from \eqref{eqn:vec_den_t1_t2_t3} as follows:
\begin{align}
    \tr{T_3} &= \tr{\sum_{s\neq t=0}^{T-1} \E\left[ A^tn_tn_s^{\top}(A^s)^{\top} A^sn_sn_t^{\top}(A^t)^{\top}\right]} \nonumber \\
    &= \sum_{s\neq t=0}^{T-1} \E\left[\tr{A^tn_tn_s^{\top}(A^s)^{\top} A^sn_sn_t^{\top}(A^t)^{\top}}\right] \nonumber\\
    &\overset{(a)}{\leq} \sum_{s\neq t=0}^{T-1} \E\left[\lambda_{\max}(A^s(A^s)^{\top})\norm{n_s}^2\tr{A^tn_tn_t^{\top}(A^t)^{\top}}\right] \nonumber\\
    & \overset{(b)}{=} \sum_{s\neq t=0}^{T-1} \lambda_{\max}(A^s(A^s)^{\top})\E\left[\norm{n_s}^2\right]\tr{A^t\E\left[n_tn_t^{\top}\right](A^t)^{\top}} \nonumber\\
    &\overset{(c)}{=} \sum_{s\neq t=0}^{T-1} d \sigma^2 \norm{A^s}^2 \sigma^2 \tr{A^t(A^t)^{\top}} \nonumber\\
    &\overset{\eqref{eqn:tr_a_d_lambda}}{\leq} d^2  \sigma^4 \sum_{s\neq t=0}^{T-1} \norm{A^t}^2 \norm{A^s}^2 \nonumber \\
    &\overset{(d)}{\leq} d^2  \tilde{\sigma}^4 \sum_{s\neq t=0}^{T-1} \norm{A^t}^2 \norm{A^s}^2. \label{eqn:t3bound} 
\end{align}
In the above steps, (a) follows from the Rayleigh-Ritz theorem (\cite{horn2012matrix}, Theorem 4.2.2), (b) follows due to independence, (c) from the noise assumptions in \eqref{eqn:noise_model}, and (d) as $\tilde{\sigma}^4 \geq \sigma^4$ due to Jensen's inequality.

Combining the results from \eqref{eqn:t1bound}, \eqref{eqn:t2bound}, and \eqref{eqn:t3bound}, we have
\begin{align*}
    \tr{\E\left[\left(x_T x_T^{\top}\right)^2\right]} &= \tr{T_1} + 2\tr{T_2} + \tr{T_3} \\
    &\leq d^2 \tilde{\sigma}^4 \left(\sum_{t=0}^{T-1} \norm{A^t}^4 + 3\sum_{s\neq t=0}^{T-1} \norm{A^t}^2 \norm{A^s}^2 \right) \\
    &\leq 3d^2 \tilde{\sigma}^4 \left(\sum_{t=0}^{T-1} \norm{A^t}^2\right)^2.
\end{align*}
Plugging this bound in \eqref{eqn:vec_den_trace_bound}, we get $$\tr{\E\left[\left(z_T z_T^{\top} - I\right)^2\right]} \leq \frac{3d^2 \tilde{\sigma}^4 \left(\sum_{t=0}^{T-1} \norm{A^t}^2\right)^2}{\left(\lambda_{\min}{(\Sigma_x)}\right)^2}.$$
Next, plugging the above bound in \eqref{eqn:vec_den_cheby}, we have 
\begin{align}
    \mathbb{P}\left( XX^{\top} \not \preceq t^2 I \right) & \leq \frac{M 3d^2 \tilde{\sigma}^4 \left(\sum_{t=0}^{T-1} \norm{A^t}^2\right)^2}{\left(\lambda_{\min}{(\Sigma_x)}\right)^2t^2} \nonumber \\
    & \overset{(a)}{=} \frac{M 3d^2 \tilde{\sigma}^4 \left(\sum_{t=0}^{T-1} \norm{A^t}^2\right)^2}{\sigma^4\left(\lambda_{\min}{(G_T)}\right)^2t^2} \nonumber \\
    & \overset{(b)}{=} C_AC_w \frac{3Md^2}{t^2}, \nonumber
\end{align}
where (a) follows as $\Sigma_x = \E[x_Tx_T^{\top}] = \sigma^2 G_T$, and, in (b), we have used the definition of ${C_A} \triangleq \left(\frac{\sum_{t=0}^{T-1} \norm{A^t}^2}{\lambda_{\min}(G_T)}\right)^2$, and $C_w \triangleq \frac{\tilde{\sigma}^4}{\sigma^4}$ from Theorem \ref{thm:vector_case}. Choosing $t = M/2$ and setting the R.H.S $\leq p/2$, we get the desired requirement for $M = 24(d^2/p)C_AC_w$. Notice that on a successful event, with probability at least $1 - p/2$, we have $XX^{\top} \preceq \frac{M^2}{4} I$, which implies $\norm{X} \leq M/2$, which further implies $$\lambda_{\min}\left(\sum_{i \in \mc{B}_j} \left( z_T^{(i)}\big(z_T^{(i)}\big)^{\top} - I \right)\right) \geq -M/2.$$ Moreover, based on \eqref{eqn:vec_den_good_error}, we conclude that $\lambda_{\min}\left(\sum_{i \in \mc{B}_j} z_T^{(i)}\big(z_T^{(i)}\big)^{\top}\right) \geq M/2$. This completes the proof of Lemma \ref{lemma:vector_den}.
\end{proof}

Combining the results from Lemma~\ref{lemma:vector_num} and Lemma~\ref{lemma:vector_den} by applying union bound, and noting that $\norm{\Sigma_x^{-1/2}} \leq 1/\sqrt{\sigma^2\lambda_{\min}(G_T)}$, we have the following bound with probability at least $1-p$:
\begin{equation}
   \norm{\hat{A}_j - A} \leq c d\sqrt{\frac{1}{p M \lambda_{\min}(G_T)}},
   \label{eqn:vector_bound_pre_boost}
\end{equation}
when $M \geq c_2 (1/p) d^2 C_w C_A$. Similar to the scalar case, the following lemma improves the above bound using the boosting step.

\begin{lemma} \label{lemma:vector_boosting} (\textbf{Vector boosting}) In \eqref{eqn:vector_bound_pre_boost}, set p = 1/8. Let $\hat{A} = \texttt{Med}(\hat{A}_1, \ldots, \hat{A}_K)$, where $\hat{A}_1, \ldots, \hat{A}_K$ are $K$ independent OLS estimators corresponding to different buckets, each satisfying \eqref{eqn:vector_bound_pre_boost}, and $\texttt{Med}$ is the geometric median operator as defined in \eqref{eqn:GM}. Given $\delta > 0$, when $K \geq \ceil{c_1 \log(1/\delta)}$, the following holds with probability at least $1 - \delta$:
\begin{equation}
    \norm{\hat{A} - A} \leq c_2 d^{3/2}\sqrt{\frac{\log(1/\delta)}{N \lambda_{\min}(G_T)}}.
    \label{eqn:vector_boosting_result}
\end{equation}
\end{lemma}
\begin{proof}
    Let $\varepsilon = c d\sqrt{\frac{1}{p M \lambda_{\min}(G_T)}}$ where $c$ is as in \eqref{eqn:vector_bound_pre_boost}. Consider a ``bad" event $\norm{\hat{A} - A} > C_\alpha \sqrt{d}\varepsilon$, where we use Lemma \ref{lemma:geometric_median} with $\alpha = 1/4$ which makes $C_\alpha = \frac{3}{2\sqrt{2}}$. Solving for $r$ by setting $C_\alpha r = C_\alpha \sqrt{d}\varepsilon$, we define the indicator random variables of events $\norm{\hat{A}_j - A}_F \geq \sqrt{d}\varepsilon$ and $\norm{\hat{A}_j - A} \geq \varepsilon$ as $Y_j$ and $Z_j$, respectively. Under this setting, the following implications hold:
    \begin{align*}
        \norm{\hat{A} - A} > C_\alpha\sqrt{d}\varepsilon &\overset{(a)}{\implies} \norm{\hat{A} - A}_F > C_\alpha\sqrt{d}\varepsilon \\
        &\overset{(b)}{\implies} \sum_{j \in [K]} Y_j > K/4 \\
        &\overset{(c)}{\implies} \sum_{j \in [K]} Z_j > K/4.
    \end{align*}
    In the above steps, (a) follows from $\norm{\hat{A} - A} \leq \norm{\hat{A} - A}_F$, (b) due to Lemma \ref{lemma:geometric_median} since we have defined the indicator random variables $Y_j$ appropriately, and (c) follows from the definition of the indicator random variables $Z_j$ and the fact that $\norm{\hat{A}_j - A}_F \leq \sqrt{d} \norm{\hat{A}_j - A}$. Based on the above implications, we have
    \begin{align*}
        \mathbb{P}\left( \norm{\hat{A} - A} > C_\alpha \sqrt{d}\varepsilon \right) \leq \mathbb{P} \left( \sum_{j \in [K]} Z_j \geq K/4 \right) \overset{(a)}{=} \mathbb{P} \left( \frac{1}{K} \sum_{j \in [K]} (Z_j - \E[Z_1]) \geq \frac{1}{4}  - \E[Z_1] \right).
    \end{align*}
    In (a), we subtracted the common expectation since the trajectories are identically distributed. Furthermore, since the collections of trajectories from different buckets are disjoint, each of the $Z_j$'s are i.i.d. random variables in $\{0,1\}$. This enables us to use Hoeffding's inequality to infer that
    \begin{align*}
        \mathbb{P} \left( \frac{1}{K} \sum_{j \in [K]} (Z_j - \E[Y_1]) \geq \frac{1}{4}  - \E[Z_1] \right) &\leq \exp{\left(-2 K (1/4 - \E[Z_1])^2\right)} \\
        & \overset{(b)}{\leq} \exp{\left(-2 K (1/4 - 1/8)^2\right)}.
    \end{align*}
    In the above steps, (b) follows due to $\E[Z_1] = \mathbb{P}(\norm{\hat{A}_1 - A} \geq \varepsilon) \leq p = 1/8.$ Based on the above, we have
    $$ \mathbb{P}\left(\norm{\hat{A} - A} > C_\alpha \sqrt{d}\varepsilon \right) \leq \exp(-K/32) \leq \delta,$$
    when $K=\ceil{32 \log(1/\delta)}$. Using this expression for $K$ in $\varepsilon= c d\sqrt{\frac{K}{p N \lambda_{\min}(G_T)}}$, where we used $M=N/K$, completes the proof.
\end{proof}
\newpage
\section{Proof of Theorem \ref{thm:corruption}: System Identification under Adversarial Corruptions}
In this section, we prove the high probability bound for system identification under adversarial corruptions, as presented in Theorem \ref{thm:corruption}. We follow an approach similar to the proof of Lemma \ref{lemma:vector_boosting} which concerns the boosting step of \texttt{Robust-SysID}. We use Lemma \ref{lemma:geometric_median} with $\alpha = 1/4$ which makes $C_\alpha = \frac{3}{2\sqrt{2}}$. Fixing $p = 1/8$, let $\varepsilon = c d\sqrt{\frac{1}{p M \lambda_{\min}(G_T)}}$ where $c$ is as in \eqref{eqn:vector_bound_pre_boost}. Consider a ``bad" event $\norm{\hat{A} - A} > C_\alpha \sqrt{d}\varepsilon$. Next, we define the following indicator random variables $\forall j \in [K]$: let $Y_j$ be the indicator random variable of the event $\norm{\hat{A}_j - A}_F \geq \sqrt{d}\varepsilon$, let $Z_j$ be the indicator random variable of the event $\norm{\hat{A}_j - A} \geq \varepsilon$, and let $W_j$ be the indicator random variable of the event ``\textit{bucket j has no corruptions}". Under this setting, we have the following implications:
\begin{align*}
    \norm{\hat{A} - A} > C_\alpha\sqrt{d}\varepsilon &\overset{(a)}{\implies} \norm{\hat{A} - A}_F > C_\alpha\sqrt{d}\varepsilon \\
    &\overset{(b)}{\implies} \sum_{j \in [K]} Y_j > K/4 \\
    &\overset{(c)}{\implies} \sum_{j \in [K]} Z_j > K/4 \\
    &\iff \sum_{j \in [K]} (Z_jW_j + Z_j(1 - W_j)) > K/4 \\
    &\overset{(d)}{\implies} \sum_{j \in [K]} (Z_jW_j + (1 - W_j)) > K/4 \\
    &\overset{(e)}{\implies} \sum_{j \in [K]} Z_jW_j + \eta N > K/4 \\
    &\iff \sum_{j \in [K]} Z_jW_j > K/4 - \eta N.
\end{align*}
In the above steps, the implications (a), (b) and (c) are justified in the proof of Lemma \ref{lemma:vector_boosting}; (d) follows as $Z_j \leq 1$ $\forall j \in [K]$, and (e) follows from the fact that $\sum_{j \in [K]} (1 - W_j) \leq \eta N$, i.e. the number of buckets with corruptions is less than the maximum number of corrupted trajectories. Based on the above implications, we have
\begin{align}
    \mathbb{P}\left( \norm{\hat{A} - A} > C_\alpha \sqrt{d}\varepsilon \right) &\leq \mathbb{P} \left( \sum_{j \in [K]}  Z_jW_j > K/4 - \eta N \right) \nonumber\\
    &= \mathbb{P} \left( \frac{1}{K} \sum_{j \in [K]} (Z_jW_j - \E[Z_jW_j]) > \frac{1}{K}\sum_{j \in [K]}\left(\frac{1}{4}  - \frac{\eta N}{K} - \E[Z_jW_j]\right) \right) \nonumber\\
    &\leq \mathbb{P} \left( \frac{1}{K} \sum_{j \in [K]} (Z_jW_j - \E[Z_jW_j]) \geq \frac{1}{4} - \frac{\eta N}{K} - p \right), \label{eqn:corruption_pre_hoeffding}
\end{align}
where the last inequality follows from the fact that $p$, based on \eqref{eqn:vector_bound_pre_boost}, is an upper-bound on the probability of an event $\norm{\hat{A}_j - A} \geq \varepsilon$ for a bucket without any corruptions. In particular, we have for all $j \in [K]$: 
\begin{align*}
    \E[Z_jW_j] &= \mathbb{P}\left( \{ \norm{\hat{A}_j - A} \geq \varepsilon\} \text{ and } \{ W_j = 1\}\right)\\
    &\overset{(a)}{\leq} \mathbb{P}\left(\{ \norm{\hat{A}_j - A} \geq \varepsilon\} \mid \{ W_j = 1\} \right)\\
    &\overset{(b)}{\leq} p,
\end{align*}
where (a) follows due to Bayes' law, and (b) due to \eqref{eqn:vector_bound_pre_boost}, where we established an estimation error bound for buckets without corruption.
%showed that the buckets without corruptions are concentrated. 
Substituting $p = 1/8$ in \eqref{eqn:corruption_pre_hoeffding}, and using Hoeffding's inequality, we have 
\begin{align*}
    \mathbb{P} \left( \frac{1}{K} \sum_{j \in [K]} (Z_jW_j - \E[Z_jW_j]) \geq \frac{1}{4} - \frac{\eta N}{K} - p \right) \leq \exp{\left(-2 K \left(1/8 - \frac{\eta N}{K}\right)^2\right)}.
\end{align*}
Setting the R.H.S $\leq \delta$ in the above display, we require $2K (1/8 - \eta N/K)^2 \geq \log(1/\delta)$, which is satisfied when $K \geq 32(\log(1/\delta) + \eta N/2)$. Furthermore, since we do not know which buckets are corrupted beforehand, we need each bucket to have $M \geq c_2 d^2 C_w C_A$ trajectories to ensure \eqref{eqn:vector_bound_pre_boost} holds for buckets without corruption. Finally, we get the requirement on the corruption fraction $\eta < 0.5/(c_1d^2C_AC_w)$ by noting that $N = MK$. This completes the proof of Theorem \ref{thm:corruption}.

\end{document}